\documentclass[11pt]{amsart}

\usepackage{amssymb}
\usepackage{amsthm}
\usepackage{amsmath}
\usepackage{graphicx}
\usepackage{amsaddr}
\usepackage{comment}
\usepackage[usenames,dvipsnames]{xcolor}
\usepackage[margin = 1in]{geometry}
\usepackage{enumerate}
\usepackage[toc,page]{appendix}
\usepackage{lineno}
\usepackage{color}
\usepackage{ulem}
\newcommand{\stkout}[1]{\ifmmode\text{\sout{\ensuremath{#1}}}\else\sout{#1}\fi}
\usepackage{hyperref}
\usepackage{mhchem}
\usepackage{siunitx}

\newcommand{\blue}{}
\definecolor{mygreen}{rgb}{0.1,0.75,0.2}

\providecommand{\bbs}[1]{\left(#1\right)}

 \newtheorem{thm}{Theorem}[section]

 \newtheorem{prop}[thm]{Proposition}

 \theoremstyle{definition}

 \theoremstyle{remark}
 \newtheorem{rem}{Remark}

 \numberwithin{equation}{section}

\newcommand{\pt}{\partial}

\newcommand{\eps}{\varepsilon}

\newcommand{\ud}{\,\,\mathrm{d}}
\newcommand{\8}{\infty}
\newcommand{\F}{\mathcal{F}}

\newcommand{\bR}{\mathbb{R}}

\newcommand{\cc}{\mathcal{C}}
\newcommand{\ssr}{\xi}
\newcommand{\Ray}{\mathcal{R}}

\newcommand{\cl}{\scriptscriptstyle{\text{CL}}}

\newcommand{\nsg}{\scriptscriptstyle{\text{SG}}}
\newcommand{\nsl}{\scriptscriptstyle{\text{SL}}}

\begin{document}

\title[Surfactant-dependent contact line dynamics]{Surfactant-dependent contact line dynamics and droplet spreading on textured substrates: derivations  and computations}

\author[Y. Gao]{Yuan Gao}
\address{Department of Mathematics, Purdue University, West Lafayette, IN}
\email{gao662@purdue.edu}

\author[J.-G. Liu]{Jian-Guo Liu}
\address{Department of Mathematics and Department of
  Physics, Duke University, Durham, NC}
\email{jliu@math.duke.edu}

\begin{abstract}
We study spreading  of a droplet, with insoluble surfactant covering its capillary surface, on a textured substrate. In this process, the surfactant-dependent surface tension dominates the behaviors of the whole dynamics, particularly the moving contact lines. This allows us to derive the full  dynamics of the droplets laid by	 the  insoluble surfactant: (i) the moving contact lines, (ii) the evolution of the capillary surface,   (iii) the surfactant dynamics on this moving surface with a  boundary condition  at the contact lines and (iv) the incompressible viscous fluids inside the droplet. Our derivations base on Onsager's principle with  Rayleigh dissipation functionals for either the  viscous flow inside droplets or the  motion by mean curvature of the capillary surface.  We also prove the Rayleigh dissipation functional for   viscous flow case is stronger than the one for the motion by mean curvature.     After  incorporating the textured substrate profile, we design a numerical scheme based on unconditionally stable explicit boundary updates and moving grids, which enable efficient computations  for many challenging examples showing  significant impacts of the surfactant to the deformation of droplets. 
\end{abstract}
\date{\today}
\maketitle
\section{Introduction}

Dynamics of droplets spreading on an impermeable substrate is not only a fundamental mathematical problem but also has  a wide range of practical applications such as droplet-based microfluidics in drug discovery, sensor design, enhanced oil recovery, surfactant replacement therapy and other dispersion technology \cite{Craster_Matar_2009, shembekar2016droplet, chou2015recent, olajire2014review}. {\blue Among those applications, surfactant, as one of the main material forms in soft matter, plays an essential role during whole spreading processes.   Other mesoscopic
constituents in  soft matter include polymers, colloids and liquid crystals; see book \cite{Doi_2013} by  \textsc{Doi}. Surfactant (i.e. surface-active agent) molecules are made of two parts,  hydrophilic part and hydrophobic part. These help surfactant to form different types of micelles, depending on  environment, so that they can either dissolve in a solvent or cover on the surface of a liquid droplet. We focus on  insoluble surfactant (known as Langmuir monolayer) in this paper. 
Surface energy, or in general interfacial energy, is very important for  flows and deformation of small liquid droplets, where the ratio between the surface area and the bulk volume is large. The addition of  surfactant  will decrease the effective surface tension of the capillary surface of a droplet if the surface energy density is convex w.r.t. the surfactant concentration, which will be explained in the next paragraph. As the insoluble surfactant spreads on the evolving capillary surface, the change of surfactant-dependent surface tension will lead to the surfactant-driven flow, such as the Marangoni flow and fingering phenomena. Most of these surfactant-driven flows are lack of mathematical validations and analysis.    Particularly, when the droplet laid by insoluble surfactant  is placed on an impermeable substrate, the dynamics of the capillary surface, the moving contact lines and the concentration of surfactant are all coupled together.  Therefore, mathematical derivations, validations and numerical simulations for dynamics of  a droplet coupled with moving contact lines are important and demanding  topics; see review article \cite{deGennes_1985} by \text{de Gennes}. }

First, the spreading process of   a small droplet placed on an impermeable textured substrate is mainly  driven by the capillary effect. That is to say, the droplet tends to minimize  the surface energy $\F$, which consists of the surface energy  of three interfaces among solid, liquid and gas. {\blue Here the surface energy density for solid-liquid interfaces (solid-gas resp.) is  denoted as $\gamma_{\nsl}$ ($\gamma_{\nsg}$ resp.)  and the surface energy density  for the liquid-gas interface without surfactant is denoted as $\gamma_0$. The variation of the total surface energy will provide the force that dominates the dynamics of small droplets.}
Now we suppose there are  insoluble surfactant  concentrating on  the evolutionary capillary surface, i.e., the interface between the liquid inside the droplet and the gas surrounding it. With the surfactant, the surface energy density on the capillary surface will depend on the surface concentration of surfactant $c$ and will be denoted as  $e(c)$.  During the spreading process,    change of the surface concentration of surfactant $c(\cdot, t)$ is induced by   stretching and evolution of the capillary surface and the surfactant also has its own convention and diffusion on the capillary surface. More importantly,  the   surfactant-dependent  surface tension $\gamma(c)$,  with the unit force/length, has the same unit with the energy density $e(c)$ (energy/area) but no longer equals $e(c)$. The relation  between the surfactant-dependent
surface tension $\gamma(c)$ and the free energy density  $e(c)$ of the surfactant-covered capillary surface  is given by $\gamma(c)=e(c)-e'(c) c;$
 see  \cite{Doi_2013, Garcke_Wieland_2006} and derivations in Section \ref{sec_energy_s}. {\blue Thus if the surface energy density $e(c)$ is convex, then $\gamma'(c)=-e''(c)c\leq 0$.  Therefore, as the surfactant disperses along the evolving capillary surface,  the   surfactant-dependent
surface tension $\gamma(c)$ will in turn significantly alter the motion of the capillary surface and moving contact lines, i.e., the lines where three phases (liquid, gas and solid) meet.  } 
This fundamental question on surfactant effect for the contact line dynamics of droplets was  discussed in the review article \cite{deGennes_1985}  by \textsc{de Gennes}.

 As mentioned above, the surfactant-dependent capillary effect dominates the whole spreading process, {\blue so the evolution of the geometric shape of the droplet coupled with the dynamics of the concentration of the insoluble surfactant on the capillary surface  are the main focus of this paper.}  We regard the geometric states, including wetting domain $\Omega_t$ and capillary surface $h(x,y, t)$, as the configuration for the droplet dynamics.  We will first derive  dynamics of the surfactant moving with the capillary surface represented by a graph function $h(x,y, t)$  with some proper boundary conditions at the contact lines. Then combining the total energy $\F$ defined in \eqref{3Denergy}, a Rayleigh dissipation functional defined in \eqref{RayQ} and  Onsager's principle \cite{Doi_2013, Doi_2021}, we derive the governing equations for the whole system. {\blue As explained below, we focus on how the surfactant-dependent surface tension $\gamma(c)$ naturally appears and dominates in the whole system. Explicitly, we will see the surfactant-dependent Laplace pressure $\gamma(c) H$ and the gradient of surfactant-dependent surface tension $\nabla_s \gamma(\cc)$ drive the motion of the capillary surface while the surfactant-dependent unbalance Young force $F_s$ drives the motion of contact lines.}
 
{\blue In the first special case that the viscosity of the fluids inside the droplets are neglected, we consider the surfactant move with the evolving capillary surface, i.e., there is no additional tangential convection w.r.t. the capillary surface for the surfactant, called ``no free-slip'' case.   In Section \ref{sec_2},} we first observe the  motion  of the capillary surface is driven by  the surfactant-dependent force $\gamma(c) H$ per unit area (known as the Laplace pressure), where $H$ is the mean curvature of the capillary surface. This observation mainly relies on the energy law \eqref{CL_energy_ex} for the capillary surface.  {\blue In this paper, we choose the convention for the mean curvature notation $H$ so that a sphere with radius $R$ in 3D has the mean curvature $H=\frac{2}{R}$.} Second,  the most complicated  competition, relaxation and balance  happen at the contact lines, so we need to derive a  surfactant-dependent unbalanced Young force at the contact lines.
Without the surfactant, the unbalanced Young force \cite{deGennes_1985} at the contact lines is
$$F_Y = \gamma_{\nsg} - \gamma_{\nsl} - \gamma_0 \cos \theta_{\cl} = \gamma_0 \bbs{\cos \theta_Y - \cos \theta_{\cl}}, \quad \cos \theta_Y:= \frac{ \gamma_{\nsg} - \gamma_{\nsl}}{\gamma_0},$$
where $\theta_{\cl}$ is the dynamic contact angle, i.e., the  angle (inside the droplet) between capillary surface and the solid substrate; see Fig. \ref{fig:ill}. 
Then with a dissipation mechanism,   Onsager's  linear response theory with friction coefficient $\ssr$, one can obtain the relation between the contact line speed $v_{\cl}$ and this driven force, and thus obtain the dynamics of the moving contact lines $\ssr v_{\cl} = F$. However, with  the presence of the surfactant, how does the surfactant transport  and how does the energy exchanges at the moving contact lines are challenging questions. 
We will first derive  a Robin-type  boundary condition \eqref{3D_bc_c} of the surfactant dynamics at the moving contact lines, which is consistent with both the mass conservation law and the energy conservation law; see Section \ref{sec_bc_3d}.  Then we adapt this boundary condition to derive the surfactant-dependent unbalanced Young force at the contact lines
\begin{equation}
F_s = \gamma_{\nsg} - \gamma_{\nsl} - \gamma(c) \cos \theta_{\cl} ,
\end{equation}
in which the surfactant-dependent surface tension is exactly the one $\gamma(c) = e(c)-e'(c) c$.  Hence the dynamics of the moving contact line with the surfactant effect is
\begin{equation}\label{MCL}
\ssr v_{\cl} = F_s.
\end{equation}
We refer to \eqref{eq_full} for the  corresponding effective Young force after including a textured substrate.

   In summary, in the special ``no free-slip'' case,   the full spreading process of the droplets can be described by (i) the continuity equation of the surfactant, (ii) the moving contact lines and (iii) the evolution of the capillary surface via curvature flow; see \eqref{3Dfull} for 3D droplets with a volume constraint and see \eqref{eq_full} for 2D droplets placed on a textured substrate including the gravitational      effect.     
       Our derivations for the geometric motion of droplets, basing on a graph representation $h(x,y, t)$ of the capillary surface,   also enable us to design an unconditionally stable and efficient numerical scheme; see Section \ref{sec_num}.   

If we further consider {\blue the general case that  there are viscous bulk fluids inside the droplet and surfactant is not only move with the capillary surface but also has ``free-slip'' with the additional tangential speed $v_s$, then we will derive the  surfactant-induced Marangoni flow inside the droplets in Section \ref{sec_Mar}.}  The derivations for the purely geometric motion in Section \ref{sec_2} can be easily adapted to the bulk  viscous flow     based on Onsager's principle with a new Rayleigh dissipation functional.
In this case,  there is an additional tangential convention  of the surfactant on the capillary surface  contributed from the bulk fluid velocity. 
This convention, together with  the surface gradient of the surfactant-dependent surface tension $\nabla_s  \gamma(\cc)$, leads to the Marangoni flow. Here $\cc$ is the surface concentration in \eqref{defCC}. {\blue Notice this additional force $\nabla_s \gamma(\cc)$ in the variation of the total surface energy exerted at the capillary surface $S_t$ is induced by the spatial-changes of the surfactant-dependent surface tension; see detailed explanations in \eqref{def_FF} and \eqref{claimG}. Therefore, this surface gradient is called Marangoni stress, and this phenomena is called Marangoni effect.} Then the Robin-type  boundary condition \eqref{3D_bc_c} for $c$ at the contact lines becomes no-flux boundary condition \eqref{bc-noflux}.
With the additional Marangoni stress, after incorporating the transport equation for the surfactant,  Onsager's principle immediately yields
the corresponding governing equations \eqref{eq_Mar} for the  surfactant-induced Marangoni flow model for droplets on a substrate; see details in Section \ref{sec_Mar}. We also show the Onsager reciprocal relations for both geometric motion case and the viscous flow case and in Proposition \ref{prop_korn}, we prove   the  dissipation functional for the viscous flow case is stronger than the one in the geometric motion model. {\blue This lower bound of the dissipation functional also helps us characterize the steady profile of the whole dynamics as a spherical cap profile with constant mean curvature while the contact angle being Young's angle; see \eqref{capS}.}

In Section \ref{sec_num}, we  propose a  numerical scheme for the full dynamics of 2D droplets laid by the surfactant and placed on a textured substrate. This unconditionally stable scheme  relies on the combination of the surfactant updates, which  constantly change the effective surface tension $\gamma(c)$,  and the splitting method with the 1st/2nd order accuracy that developed   in  \cite{gao2020gradient} for the purely  geometric motion of  a single droplet without  surfactant. Specifically, at each step, we first use  unconditionally stable explicit updates for the moving  contact lines, which efficiently decouple the computations for the motion  of the capillary surface and the contact line dynamics. Then we adapt the arbitrary Lagrangian-Eulerian (ALE) method to handle  the moving grids to update the profile of the  capillary surface and the concentration of surfactant with  Robin-type boundary condition \eqref{BC_c} at the contact lines. Based on this, some challenging examples showing  significant effects of surfactant to the droplets dynamics will be conducted in Section \ref{sec_exm}. These include  (i) a surface tension decreasing phenomena and  asymmetric capillary surfaces due to  presence of surfactant; (ii) an enhanced rolling down for droplets placed on an inclined substrate;  (iii) droplets on a textured substrate or a container with different surfactant concentrations.

{\blue We incompletely list some recent theoretical and numerical studies on this subject, including droplets with insoluble or soluble surfactant.
 The  lubrication approximation for the 
thin film covered by
insoluble surfactant are investigated by \textsc{Garcke and Wieland} in 
\cite{Garcke_Wieland_2006}. They also proved the global existence and positivity for the solution to the resulting thin film equation coupled with transport of insoluble surfactant. However, the contact line dynamics was not considered in \cite{Garcke_Wieland_2006}.  
 Some numerical methods for computing the droplet dynamics coupled with moving contact lines and insoluble surfactant are developed; see \cite{lai2008immersed, Lai_2010} for the immersed boundary method, see \cite{xu2014level, Zhang_Xu_Ren_2014} for the level set method and see \cite{Ganesan_2015} for an  arbitrary Lagrangian–Eulerian finite element method. There are many other  studies on the modeling and measurement of the surfactant enhancement for  spreading and evaporation of  droplets in various physical situations; c.f. \cite{Karapetsas_Craster_Matar_2011, Karapetsas_Sahu_Matar_2016}.  We refer to \cite{chen2014conservative, gao2017global, wu2019drying} for droplets or thin film involved dynamics coupled with soluble surfactant. Finally, for  general derivation methods for complex fluids via Onsager's principle, we refer to \textsc{Wang, Qian and Sheng} \cite{Qian_Wang_Sheng_2006}
and a recent review article by \textsc{Doi} \cite{Doi_2021}. 
 }

The organization of this paper is as follows. In Section \ref{sec_2}, we derive the geometric motion of 3D droplets, i.e., the moving contact lines, the evolution of the capillary surface and the surfactant dynamics on it, in which we incorporate  the surfactant-dependent surface tension $\gamma(c)$. In Section \ref{sec_Mar}, we derive the full dynamics of a 3D droplet with the surfactant-induced Marangoni flow inside it.  In Section \ref{sec_num}, we present the numerical scheme for 2D droplets placed on an inclined textured substrate based on the splitting method. In Section \ref{sec_exm}, we conduct some challenging examples showing the significant contributions  of the surfactant to the whole spreading process.

\section{Derivation for 3D contact line dynamics with surfactant }\label{sec_2}
We study the motion of a 3D droplet placed on a substrate, which is identified by the region $A_t:=\{(x,y, z);~ (x,y)\in\Omega_t, \, 0\leq z\leq  h(x,y,t)\}$ with a sharp interface. The motion of this droplet is described by a moving capillary surface $S_t$, and a partially wetting domain $\Omega_t$ with a free boundary $\pt \Omega_t$ (physically known as the contact lines);  see Fig. \ref{fig:ill}(a). To clarify notations, 
let $\Omega_t$ be a wetting domain, which is a simply connected 2D open set. Let $h(x,y, t), (x,y)\in \Omega_t$ be the graph representation  for the moving capillary surface. Then the capillary surface can be represented as
\begin{equation}
S_t:= \{(x,y,h(x,y,t),\quad (x,y)\in \Omega_t\}.
\end{equation}

Denote  $\gamma_{\nsl}$ ($\gamma_{\nsg}$ resp.) as the interfacial surface energy density between solid-liquid phases (solid-gas resp.). $\gamma_{\nsl}$, $\gamma_{\nsg}$ are constants but the interfacial surface energy on the capillary surface, i.e., the interface between liquid and gas,  will depend on the insoluble surfactant on it.  {\blue Denote $\cc(x,y,z,t)$  as  the surface concentration of the surfactant 
 on the capillary surface, i.e., the number of the surfactant molecules per unit area\footnote{We remark the standard notation in chemistry for the surface  concentration is $\Gamma$.}. 
 Denote 
 \begin{equation}\label{defCC}
 c(x,y,t):=  \cc(x,y, h(x,y,t),t)
 \end{equation}
 as the ``concentration'' in terms of measure $\sqrt{1+|\nabla h|^2} \ud x \ud y$ on $\Omega_t$, hence the concentration in surface element satisfies
 \begin{equation}
 \cc \, \mathcal{H}^2(S_t) = c \sqrt{1+|\nabla h|^2} \ud x \ud y.
 \end{equation}
 Here $\mathcal{H}^2(S_t)$ is the Hausdorff measure of $S_t$. 
 We will compute the transport of surfactant by using this ``concentration'' $c(x,y,t)$.  This derivation in terms of $(x,y)\in \Omega_t$ is much simpler and is equivalent to the transport of $\cc(x,y,z,t)$ on the moving surface $S_t$; see Proposition \ref{prop_equiv} below. }
 
Let $e(c)$ be the  surface energy density on the capillary surface.
Then the total surface energy  of the droplet is
\begin{equation}\label{energy}
\F(h(x,y,t), \Omega_t, c(x,y, t)):= \int_{\Omega_t}  e(c) \sqrt{1+ |\nabla h|^2} \ud x \ud y + (\gamma_{\nsl}-\gamma_{\nsg}) \int_{\Omega_t} \ud x \ud y.
\end{equation}

We assume the following two constraints: the volume constraint $V$ for the droplet and the total mass constraint $M_0$ for the surfactant, i.e. 
\begin{equation}\label{const}
\begin{aligned}
\int_{\Omega_t}  h \ud x \ud y = V,\qquad 
\int_{\Omega_t}  c \sqrt{1+|\nabla h|^2} \ud x = M_0.
\end{aligned}
\end{equation}
Define the contact angles (the angle  inside the droplet $A$ between the capillary surface and the solid substrate) at contact lines $\pt \Omega$ as $\theta_{\cl}$ such that
\begin{equation}\label{con-ang}
\tan \theta_{\cl} = |\nabla h|.
\end{equation}
\begin{figure}
\includegraphics[scale=0.5]{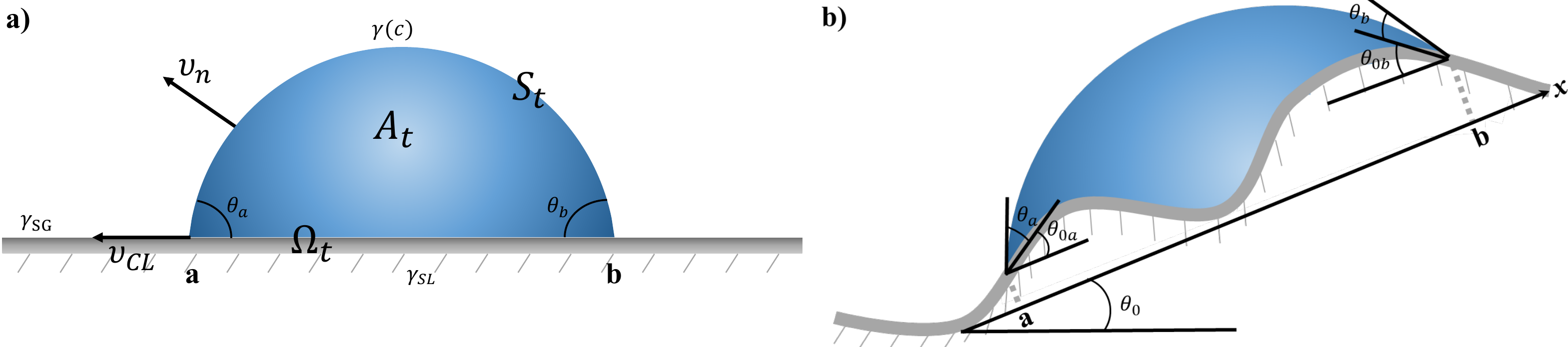} 
\caption{Illustration of surface tensions $\gamma_{\nsg}, \gamma_{\nsl}, \gamma(c)$ on three interfaces,  contact angle $\theta_{\cl}$, capillary surface $S_t$ and wetting domain $\Omega_t$ for droplets on a plane (left) or on  an inclined substrate with an effective inclined angle $\theta_0$ (right).)  }\label{fig:ill}
\end{figure}

We regard the geometric states, i.e., the wetting domain $\Omega_t$ and the capillary surface $h(x,y, t)$ as  a configuration of  droplets. We assume the surfactant concentrates on the capillary surface, move with the evolutionary surface and may also has its own convention and diffusion on the surface.
In the following subsections, we will first {\blue give some kinematic descriptions for the moving surface and  the surfactant concentration on the capillary surface. Then we impose a boundary condition for the concentration of the surfactant $c$ on a prescribed moving capillary surface $h$ to preserve total mass and based on it we compute the rate of change of total surface energy; see Section \ref{sec2.1_K}.   
The rate of change of total energy $\dot{\F}$ in \eqref{final_rate} quantifies the work done by the open system (capillary surface laid by surfactant and its contact line) against friction \cite{goldstein2002classical}. Then we only need to determine the velocity fields, including normal velocity of the capillary surface $v_n$ and contact line speed $v_{\cl}$, via Onsager's principle.
In Section \ref{sec2.2_E},  from the energetic considerations, we introduce a specific  Rayleigh dissipation functional and using Onsager's principle to derive the contact line motion driven by the surfactant-dependent unbalanced Young force and finally the governing equations for the full dynamics of 3D droplets laid by the surfactant.}

\subsection{Kinematic descriptions for the moving surface,  the surfactant concentration and energy}\label{sec2.1_K}
Surfactant dynamics on an evolutionary surface with the mass conservation law is a well-known model, c.f. \cite{Stone_1990}. For the case the surface  has a graph representation $h(x,y, t), (x,y)\in\Omega_t$, we will provide a simple kinematic descriptions for the moving capillary surface, and the continuity equation for the concentration of the surfactant $c(x,y, t), (x,y)\in\Omega_t$ on the capillary surface. Based on this, the rate of change of the free energy $\F$ will then be calculated.

\subsubsection{The continuity equation for the surfactant  represented in the $xy$-plane}
{\blue
First, we describe the motion of the capillary surface $S_t$. Given a capillary surface with a graph representation $h(x,y, t), (x,y)\in\Omega_t$, any point on this moving capillary surface can be represented as
\begin{equation}
X(t) = (x(t), y(t), h(x(t), y(t), t)).
\end{equation}
Then the observed velocity of this point is
\begin{equation}
\dot{X} = (\dot{x}, \dot{y}, h_x \dot{x}+h_y \dot{y} + h_t).
\end{equation}
We assume there is an underlying velocity field $v\in \bR^3$ driving the motion of the capillary surface, i.e., 
\begin{equation}\label{particleV}
\dot{X}(t)= v(X(t),t).
\end{equation}
To clarify notations for functions of $(x,y)$ and functions of $(x,y,z)$, we introduce notations
\begin{equation}
\begin{aligned}
&v(x,y,z,t)=:(v_x, v_y, v_z)(x,y,z,t),\\
& v(x,y,z,t)\big|_{z=h(x,y,t)} =: (v_1, v_2, v_3)(x,y,t).
\end{aligned}
\end{equation}
Using these notations,  \eqref{particleV} implies the evolution of the surface  in terms of $(x,y)\in\Omega_t$  
\begin{equation}\label{eve}
h_t + v_1 h_x + v_2 h_y = v_3.
\end{equation}
Denote the normal vector as $n := \frac{1}{\sqrt{1+|\nabla h|^2}}(- h_x, - h_y, 1)$ and the tangential vectors as $\tau_1 := (1, 0, h_x), \,\, \tau_2 := (0,1, h_y)$.
Then \eqref{eve} becomes
\begin{equation}\label{velocityNN}
v_n := v \cdot n = \frac{h_t}{\sqrt{1+|\nabla h|^2}}.
\end{equation}
 Now we  express the velocity in the directions of the normal vector $n$ and the tangential vectors $\tau_1,\, \tau_2$ as
\begin{equation}\label{velocity}
\begin{aligned}
&(v_1, v_2, v_3)(x,y,t)= (v_n n +  f \tau_1 + g \tau_2)(x,y,t),
\end{aligned}
\end{equation}
where  $\displaystyle f(x,y,t):=\frac{v\cdot \tau_1}{|\tau_1|^2},\,\, g(x,y,t):=\frac{v\cdot \tau_2}{|\tau_2|^2}$.}

Second, we describe the dynamics of the concentration of the insoluble surfactant on the moving surface.
Recall $\cc(x,y,z, t), (x,y,z)\in S_t$ is the surface concentration of surfactant on the capillary surface  and \eqref{defCC}. Then using \eqref{eve},   we have
\begin{align}
&\frac{\ud}{\ud t} \cc(x(t), y(t), h(x(t),y(t),t),t) = (\pt_t + v \cdot \nabla )\cc = (\pt_t +
v_1 \pt_x + v_2 \pt_y )c =  \frac{\ud}{\ud t} c( x(t), y(t),t). \label{rel_c}
\end{align}
Therefore, the dynamics of the surfactant can be fully described by $c(x,y,t)$ on the $xy$-plane, as explained below.

We derive the continuity equation for $c(x,y,t), (x,y)\in\Omega_t.$ To do so, we use the $xy$-component of velocity $v$ to define a flow map on the $xy$-plane
\begin{align}\label{flowmap}
\left\{ \begin{array}{cc}
\dot{x} =v_1 = v_n n_1 + f = \frac{-h_x v_n}{\sqrt{1+|\nabla h|^2}} + f,\\
\dot{y} = v_2 = v_n n_2 + g = \frac{-h_y v_n}{\sqrt{1+|\nabla h|^2}} + g.
\end{array} \right.
\end{align}
{\blue This flow map defines a  2D moving surface element $\omega_t \subset \Omega_t$ via   $\omega_t =\{(x(t), y(t)); (x_0,y_0)\in \omega_0\}$ 
with any given initial surface element $\omega_0$. In the absence of diffusion, $\omega_t$ can be regarded as a material element.  That is to say, the  mass in the material element $\omega_t$ is conserved
\begin{equation}\label{mass_c}
 \frac{\ud }{\ud t} \int_{\omega_t} c \sqrt{1+ |\nabla h|^2} \ud x \ud y=0.
\end{equation}
}

By \eqref{mass_c} and  the Reynolds transport theorem
\begin{equation}
0 = \frac{\ud }{\ud t} \int_{\omega_t} c \sqrt{1+ |\nabla h|^2} \ud x \ud y
= \int_{\omega_t} \pt_t  \bbs{c \sqrt{1+ |\nabla h|^2}} + \nabla \cdot \bbs{c \sqrt{1+|\nabla h|^2} \left( \begin{array}{c}
v_1\\
v_2
\end{array} \right) } \ud x \ud y.
\end{equation}
Then by the arbitrary of $\omega_t$, the continuity equation for $c(x,y,t)$ is
\begin{equation}\label{eq-con}
\pt_t  \bbs{c \sqrt{1+ |\nabla h|^2}} + \nabla \cdot \bbs{c \sqrt{1+|\nabla h|^2} \left( \begin{array}{c}
v_1\\
v_2
\end{array} \right) }=0 \quad \text{ in } \Omega_t.
\end{equation}
Plugging $v_1, v_2$ defined in \eqref{flowmap}, we obtain the continuity equation for $c$
\begin{equation}
\begin{aligned}
0 =& (\pt_t c) \, {\sqrt{1+|\nabla h|^2}}  + \frac{c}{\sqrt{1+|\nabla h|^2}} \nabla h \cdot \nabla h_t - \nabla \cdot \bbs{\frac{ c h_t}{\sqrt{1+|\nabla h|^2}} \nabla h}  +  \nabla \cdot \bbs{c \sqrt{1+|\nabla h|^2} \left( \begin{array}{c}
f\\
g
\end{array} \right) } \\
=& (\pt_t c) \, {\sqrt{1+|\nabla h|^2}}   - h_t \nabla \cdot \bbs{\frac{ c}{\sqrt{1+|\nabla h|^2}} \nabla h} +  \nabla \cdot \bbs{c \sqrt{1+|\nabla h|^2} \left( \begin{array}{c}
f\\
g
\end{array} \right) } 
\end{aligned}
\end{equation}
After simplification, the continuity equation for $c$ is
\begin{equation}\label{eq-con10}
\begin{aligned}
  \pt_t c- v_n \nabla c \cdot  \frac{ \nabla h}{\sqrt{1+|\nabla h|^2}}   +  v_n c H+ \frac{1}{ {\sqrt{1+|\nabla h|^2}} } \nabla \cdot \bbs{c \sqrt{1+|\nabla h|^2} \left( \begin{array}{c}
f\\
g
\end{array} \right) } =0
\end{aligned}
\end{equation}
for $(x,y)\in \Omega_t$, where $H:=-\nabla \cdot \bbs{\frac{\nabla h}{\sqrt{1+|\nabla h|^2}}}$ is the mean curvature.

\subsubsection{Comparison with the lift-up dynamics of $\cc$ on the capillary surface}\label{sec_lift}
We now compare the continuity equation \eqref{eq-con} for $c(x,y,t)$ on $\Omega_t$ with the original 3D concentration $\cc(x,y,z,t)$ on $S_t$.
We have the following proposition on the equivalent formulation  of the continuity equation in terms of $\cc$ defined on the moving surface $S_t$. The proof of this proposition will be given in Appendix \ref{appA}.
\begin{prop}\label{prop_equiv}
The continuity equation \eqref{eq-con} can be recast as
\begin{equation}\label{eq-sur}
\begin{aligned}
 (\pt_t  + v \cdot \nabla ) \cc +\cc\nabla_s \cdot  v_s + v_n \cc H=0 \quad \text{ on } S_t,
\end{aligned}
\end{equation}
where $H=\nabla_s \cdot n$, $\nabla_s$ is the surface divergence and {\blue $v_s$ is the tangent velocity
$
v_s = v - (v\cdot n) n.$
}
\end{prop}
\begin{rem}
By some elementary calculations, we remark 
\eqref{eq-sur} is equivalent to \cite[(2.13)]{Garcke_Wieland_2006} and also equivalent to \cite[(2.10)]{Lai_2010}.
Indeed, the tangential  convection can be combined with the last two terms in \eqref{eq-sur} in a conservative form, i.e.,
\begin{equation}
\begin{aligned}
0=&\pt_t \cc + v_n n \cdot \nabla \cc + v_s \cdot \nabla \cc+ \cc\nabla_s \cdot  v_s + v_n \cc H\\
=& \pt_t \cc + v_n n \cdot \nabla \cc + \nabla_s \cdot(\cc v_s) + v_n \cc H
\\
=& \pt_t \cc + v_n n \cdot \nabla \cc + \nabla_s \cdot(\cc v), \qquad \text{\cite[(2.13)]{Garcke_Wieland_2006}}.
\end{aligned}
\end{equation}
Notice also the last two terms can be combined together, so \eqref{eq-sur} is also equivalent to 
\begin{equation}
\begin{aligned}
0=&\pt_t \cc + v \cdot \nabla \cc + \cc\nabla_s \cdot  v_s +  v_n \cc H\\
=& (\pt_t  + v \cdot \nabla) \cc + \cc\nabla_s \cdot  v, \qquad  \text{\cite[(2.10)]{Lai_2010}}.
\end{aligned}
\end{equation}
{\blue We point out all these equivalent equations  differ  from ones presented in \textsc{Stone} \cite[(6)]{Stone_1990}, i.e.,
\begin{equation}
\pt_t \cc + \stkout{v_n n \cdot \nabla \cc} + \nabla_s \cdot(\cc v_s) + v_n \cc H=0.
\end{equation}
  The second term above $v_n n \cdot \nabla \cc$ vanishes   only  if concentration $\mathcal{C}$ has a  constant normal extension outside the moving surface \cite{cermelli2005transport}.} However, from  \textsc{Gurtin} \cite{gurtin1999configurational}, $(\pt_t  + v \cdot \nabla)$ is a tangential derivative of the space-time surface $\displaystyle \cup_{t\geq 0} S_t \times \{t\}$, so there is no need to extend $\cc$ outside the space-time surface.
\end{rem}

{\blue To describe the evolution of the capillary surface, we only need the normal velocity $v_n$ of the fluids. In the first special case, we consider the surfactant move with the evolving capillary surface, i.e., there is no additional tangential convection w.r.t. the capillary surface for the surfactant. We will call this special case, as ``no free-slip'' case. In this case, the continuity equation \eqref{eq-con10} can be completely described via $v_n$ and becomes 
\begin{equation}\label{eq_con18}
\begin{aligned}
\pt_t c  - v_n  \nabla \cdot \bbs{\frac{ c}{\sqrt{1+|\nabla h|^2}} \nabla h} = & \pt_t c- v_n \nabla c \cdot  \frac{ \nabla h}{\sqrt{1+|\nabla h|^2}}   +  v_n c H=0.
\end{aligned}
\end{equation}
This formula is particularly efficient for simulating the purely geometric motion of the droplet and the surfactant is pinned to move with the capillary surface. For the general case that the continuity equation is completely described via $v$, one needs to consider the fluids inside the droplets instead of the purely geometric motion; see Section \ref{sec_Mar}.
}

\subsubsection{Diffusion of  surfactant on the evolutionary surface}
Furthermore, from some elementary calculations, the Dirichlet energy for the surfactant on the capillary surface is
\begin{equation}
 \frac12 \int_{S_t} |\nabla_s \cc|^2 \ud s = \frac12 \int_{\Omega_t} \frac{1}{\sqrt{1+|\nabla h|^2}} \nabla c \cdot (M \nabla c) \ud x \ud y,
\end{equation}
where $M:
=I+\bbs{\begin{array}{c}
-h_y\\ h_x
\end{array}} \bbs{-h_y, h_x}$.
Then the variation of the Dirichlet energy gives  the Laplace–Beltrami operator in the graph representation,
$$
\Delta_s c := \frac{1}{\sqrt{1+|\nabla h|^2}} \nabla \cdot  \bbs{\frac{1}{\sqrt{1+|\nabla h|^2}}  M \nabla  c}.
$$
Thus in the ``no free-slip'' case, the continuity equation \eqref{eq_con18}  for the  surfactant  with additional diffusion becomes
\begin{equation}\label{3D_ceq}
c_t - \frac{h_t}{\sqrt{1+|\nabla h|^2}} \nabla \cdot \bbs{\frac{c \nabla h}{\sqrt{1+|\nabla h|^2}}}=D \Delta_s c,
\end{equation}
which is equivalent to
\begin{equation}\label{tm_mul}
\pt_t c- v_n \nabla c \cdot  \frac{ \nabla h}{\sqrt{1+|\nabla h|^2}}   +  v_n c H = D \Delta_s c.
\end{equation}
Here $D>0$ is a diffusion constant.
This is the continuity equation with diffusion for the surfactant dynamics on the moving surface and we will impose the no-flux boundary condition for \eqref{3D_ceq} below.
Another equivalent form of \eqref{3D_ceq} in the conservative form  is
\begin{equation}\label{3D_c_con}
 \pt_t \bbs{c {\sqrt{1+|\nabla h|^2}} } - \nabla \cdot \bbs{\frac{ c h_t}{\sqrt{1+|\nabla h|^2}} \nabla h}   =  D  \nabla \cdot  \bbs{\frac{1}{\sqrt{1+|\nabla h|^2}}  M \nabla  c}.
\end{equation}

\subsubsection{The rate of change of the energy  on the capillary surface}\label{sec_energy_s}
In this section, given an evolutionary capillary surface $h(x,y,t)$ and the associated surfactant dynamics \eqref{3D_ceq} with \eqref{3D_bc_c}, we calculate the rate of change of the energy   for the capillary surface.

Consider the free energy on the capillary surface
\begin{equation}
\F_0=\int_{\Omega_t}  e(c) \sqrt{1+ |\nabla h|^2} \ud x \ud y,
\end{equation}
where $e(c)$ is the energy density on the capillary surface and $c(x,y,t)$ satisfies \eqref{3D_ceq}.

However, in  calculations of the rate of change of the energy, the work done  by the  surface tension per unit time shall be a surfactant-dependent one, given by $\gamma(c) H v_n$, where $\gamma(c)$ is the effective surface tension and $H=-\nabla \cdot \bbs{\frac{\nabla h}{\sqrt{1+|\nabla h|^2}}}$ is the mean curvature.  We will derive this energy conservation law  below.

First, the   relation between the surfactant-dependent
surface tension $\gamma(c)$ with the free energy density  $e(c)$ is given by  \cite{Doi_2013, Garcke_Wieland_2006}  
\begin{equation}\label{gg}
\gamma(c)= e(c)-e'(c) c.
\end{equation}
Indeed, from \cite[(4.32)]{Doi_2013}, 
$\gamma(c) = \gamma_0 - \Pi_A(c)$. Here $\gamma_0=e(0)$ and $\Pi_A$ is the surface pressure which can be calculated by the osmotic pressure inside the interfacial layer due to the inhomogeneous surface concentration of surfactant. From the same derivations as  \cite[(2.23)]{Doi_2013}, 
$$ \Pi_A(c)= -e(c) + e'(c) c + e(0).$$
Thus we have \eqref{gg}.

Second, multiplying \eqref{tm_mul} by $e'(c)\sqrt{1+|\nabla h|^2}$, we have
\begin{align*}
\bbs{\pt_t e(c)}  \sqrt{1+|\nabla h|^2}- \nabla e(c) \cdot   \frac{ h_t \nabla h }{ \sqrt{1+|\nabla h|^2}} +h_t  H e'(c) c = D \sqrt{1+|\nabla h|^2}  e'(c) \Delta_s  c.
\end{align*}
We can recast this in the conservative form
\begin{equation}\label{CL_energy}
\begin{aligned}
&\pt_t \bbs{e(c)  \sqrt{1+|\nabla h|^2}}- \nabla \cdot \bbs{ e(c) \cdot   \frac{ h_t \nabla h }{ \sqrt{1+|\nabla h|^2}}} + h_t  H e'(c) c\\
&\qquad \qquad - e(c) \pt_t \bbs{ \sqrt{1+|\nabla h|^2}} + e(c) \nabla \cdot \bbs{ \frac{ h_t \nabla h }{ \sqrt{1+|\nabla h|^2}}  } 
 = D \sqrt{1+|\nabla h|^2}  e'(c) \Delta_s  c.
\end{aligned}
\end{equation}
Then using the identity 
\begin{equation}\label{MC_ID}
-\pt_t \bbs{ \sqrt{1+|\nabla h|^2}} +  \nabla \cdot \bbs{ \frac{ h_t \nabla h }{ \sqrt{1+|\nabla h|^2}}  }  = -h_t H 
\end{equation}
and relation \eqref{gg}, we simplify \eqref{CL_energy} as
\begin{equation}\label{CL_energy_ex}
\begin{aligned}
&\pt_t \bbs{e(c)  \sqrt{1+|\nabla h|^2}}- \nabla \cdot \bbs{ e(c) \cdot   \frac{ h_t \nabla h }{ \sqrt{1+|\nabla h|^2}}}  - \gamma(c)   Hh_t  
 = D \sqrt{1+|\nabla h|^2}  e'(c) \Delta_s  c.
\end{aligned}
\end{equation}
The first term in \eqref{CL_energy_ex} is the rate of change of the energy density per unit time per unit area in the $xy$-plane. The second term in \eqref{CL_energy_ex} is the flux of energy density. The third term in \eqref{CL_energy_ex}, i.e., $-\gamma(c) H h_t = -\gamma(c) H v_n \sqrt{1+|\nabla h|^2}$ is the rate of work done by the surface tension per unit area in the $xy$-plane. The last term in \eqref{CL_energy_ex} is the energy density  dissipation due to the diffusion of the surfactant.

Now we focus on our goal to derive the contact line dynamics which is driven by the effective surface tension $\gamma(c)$. {\blue However, notice the boundary condition can not be derived by taking trace of the interior velocity, so we will apply the Reynolds transport theorem using the boundary velocity of the moving domain. } In the next subsection, we will discuss  the correct boundary condition for the concentration of the surfactant.

The fundamental relation \eqref{gg}
 implies the decreasing of $\gamma(c)$ from the convexity of energy density $e(c)$. {\blue Many interesting physical phenomena can be explained by  the gradient of effective surface tension $\gamma(c)$ due to the gradient of surfactant concentration. For instance, the gradient of surface tension will drive the spreading of the droplet and the surfactant rapidly aggregates to the direction of lower concentration, which  forms Marangoni flow on the surface. Particularly, when the concentration of surfactant is above a critical micelle value, the surfactant-driven spreading will lead to fingering phenomenon, an unstable structure on surfactant-laden droplets \cite{Craster_Matar_2009}. }
We remark a typical $\gamma$, derived from the Langmuir equation, is given by 
\begin{equation}\label{Lang}
\gamma(c) = \gamma_0 + c_s kT \ln (1-\frac{c}{c_s}),
\quad e(c) = \gamma_0 + kT \bbs{(c_s -c)\ln(c_s-c) + c \ln c - c_s \ln c_s},
\end{equation}
where $\gamma_0$ is the surface tension without  the surfactant and $c_s$ is the saturated concentration \cite{Doi_2013}. We will demonstrate numerical examples using this typical $\gamma$.

\subsubsection{Impose a boundary condition at the contact line for the surfactant to preserve total mass}\label{sec_bc_3d}
To derive a boundary condition for the surfactant equation at the contact lines, we now  apply the Reynolds transport theorem for the whole wetting domain $\Omega_t$ up to its boundary. We shall be careful when using the boundary velocity in the Reynolds transport theorem because the boundary condition of a PDE can not be derived by taking trace of the interior equation.

Denote $n_{\ell}$ as the outer normal of the contact line $\pt \Omega_t$ in the $xy$-plane and $v_{\ell}$ as the velocity of the contact line. We have $n_\ell = -\frac{\nabla h}{|\nabla h|}$ on the contact line $\pt \Omega_t$. Denote the normal speed of the contact line as $v_{\cl} := v_\ell \cdot n_\ell$.   By the Reynolds transport theorem   we have
\begin{equation}\label{tm_mass}
\begin{aligned}
&\frac{\ud }{\ud t} \int_{\Omega_t} c \sqrt{1+|\nabla h|^2} \ud x \ud y =\int_{\Omega_t} \pt_t \bbs{c  \sqrt{1+|\nabla h|^2}  } \ud x \ud y + \int_{\pt\Omega_t} c  \sqrt{1+|\nabla h|^2}  v_\ell\cdot  n_\ell \ud s. 
\end{aligned}
\end{equation}
Then by \eqref{3D_c_con} and integration by parts, we obtain
\begin{equation}\label{2.24}
\int_{\Omega_t} \pt_t \bbs{c  \sqrt{1+|\nabla h|^2}  } \ud x \ud y = D \int_{\pt \Omega_t}   \frac{1}{\sqrt{1+|\nabla h|^2}}   M \nabla  c \cdot n_\ell \ud s+ \int_{\pt \Omega_t} \frac{c h_t}{\sqrt{1+|\nabla h|^2}} \nabla h  \cdot n_\ell  \ud  s
\end{equation}
Notice the definition of $M$ gives
\begin{equation}\label{M}
M \nabla c \cdot n_\ell = n_\ell \cdot \nabla c + n_\ell \cdot (-h_y, h_x) ~(-h_y, h_x) \cdot \nabla c = n_\ell \cdot \nabla c \quad  \text{ on }\pt \Omega_t. 
\end{equation}
Notice also the compatibility condition $\frac{\ud h(x(t), y(t), t)}{\ud t}=0$ on the contact line gives
\begin{equation}\label{com_21}
h_t = - \nabla h \cdot v_\ell.
\end{equation} 
Then from   $n_\ell = -\frac{\nabla h}{|\nabla h|}$ on  $\pt \Omega_t$,  \eqref{com_21}
becomes
\begin{equation}\label{com_22}
h_t =  |\nabla h| n_\ell \cdot v_\ell = |\nabla h| v_{\cl},
\end{equation}
which implies
\begin{equation}\label{useID1}
 \frac{ h_t}{\sqrt{1+|\nabla h|^2}} \nabla h  \cdot n_\ell =  -  \frac{|\nabla h|^2 }{\sqrt{1+|\nabla h|^2}}  v_{\cl}.
\end{equation}
Thus \eqref{2.24} can be further simplified as
\begin{align}
\int_{\Omega_t} \pt_t \bbs{c  \sqrt{1+|\nabla h|^2}  } \ud x \ud y 
 =  \int_{\pt \Omega_t} D \frac{n_\ell \cdot \nabla c  }{\sqrt{1+|\nabla h|^2}} -  c\frac{|\nabla h|^2 }{\sqrt{1+|\nabla h|^2}}  v_{\cl} \ud s.
\end{align}
Plugging this into \eqref{tm_mass}, we have
\begin{equation}\label{tm223}
\begin{aligned}
&\frac{\ud }{\ud t} \int_{\Omega_t} c \sqrt{1+|\nabla h|^2} \ud x \ud y\\
=&  \int_{\pt \Omega_t} \frac{D}{\sqrt{1+|\nabla h|^2}} n_\ell \cdot \nabla c - \int_{\pt \Omega_t} c\frac{|\nabla h|^2 }{\sqrt{1+|\nabla h|^2}}  v_{\cl}- c \sqrt{1+|\nabla h|^2} v_{\cl}  \ud s\\
=& \int_{\pt \Omega_t} \frac{D}{\sqrt{1+|\nabla h|^2}} n_\ell \cdot \nabla c + \int_{\pt \Omega_t} c\frac{1}{\sqrt{1+|\nabla h|^2}} v_{\cl}  \ud s.
\end{aligned}
\end{equation}

Therefore,
in order to maintain the mass conservation law
\begin{equation}
0=\frac{\ud }{\ud t} \int_{\Omega_t} c \sqrt{1+|\nabla h|^2} \ud x \ud y =   \int_{\pt \Omega_t}   \frac{1}{\sqrt{1+|\nabla h|^2}} \bbs{   D  \, n_\ell\cdot \nabla  c  +  c \, v_{\cl} }  \ud s,
\end{equation}
we impose the following Robin boundary condition for \eqref{3D_ceq}
\begin{equation}\label{3D_bc_c}
 D \, n_\ell \cdot \nabla c + c\,  v_{\cl} = 0 \quad \text{ on }\pt \Omega_t.
\end{equation}

\subsubsection{Rate of change of the total surface energy}
Using the  Reynolds transport theorem for the surface energy
\begin{equation}\label{D_energy1}
\begin{aligned}
 &\frac{\ud}{\ud t}\int_{\Omega_t} e(c) \sqrt{1+|\nabla h|^2} \ud x \ud y \\
 =& \int_{\Omega_t} \pt_t \bbs{e(c) \sqrt{1+|\nabla h|^2} }  \ud x \ud y + \int_{\pt \Omega_t}e(c) \sqrt{1+|\nabla h|^2} v_{\cl} \ud s\\
= & \int_{\Omega_t}    \nabla \cdot \bbs{ e(c) \cdot   \frac{ h_t \nabla h }{ \sqrt{1+|\nabla h|^2}}}  + \gamma(c)   Hh_t +D e'(c)  \nabla \cdot \bbs{\frac{1}{\sqrt{1+|\nabla h|^2}}   M \nabla c }  \ud x \ud y \\
& + \int_{\pt \Omega_t}e(c) \sqrt{1+|\nabla h|^2} v_{\cl} \ud s ,
 \end{aligned}
 \end{equation}
 where we used \eqref{CL_energy_ex} in the last equality.
 Then using the integration by parts, \eqref{D_energy1} becomes
 \begin{equation}\label{D_energy11}
\begin{aligned}
 &\frac{\ud}{\ud t}\int_{\Omega_t} e(c) \sqrt{1+|\nabla h|^2} \ud x \ud y \\
 =& \int_{\Omega_t} h_t H \gamma(c)   -D \frac{e''(c) }{\sqrt{1+|\nabla h|^2}}   \nabla c \cdot M \nabla  c \ud x \ud y\\
 &+  \int_{\pt \Omega_t} e(c)\frac{ h_{t}}{\sqrt{1+|\nabla h|^2}} \nabla h \cdot n_\ell + D e'(c) \frac{M \nabla c \cdot n_{\ell}}{\sqrt{1+|\nabla h|^2}} \ud s 
& + \int_{\pt \Omega_t}e(c) \sqrt{1+|\nabla h|^2} v_{\cl} \ud s.
 \end{aligned}
 \end{equation}
 Then using \eqref{useID1} and \eqref{M}, by the same calculations as \eqref{tm223}, we have
  \begin{equation}\label{D_energy2}
\begin{aligned}
 &\frac{\ud}{\ud t}\int_{\Omega_t} e(c) \sqrt{1+|\nabla h|^2} \ud x \ud y\\
  =&\int_{\Omega_t} h_t \gamma(c) H - D \frac{e''(c) }{\sqrt{1+|\nabla h|^2}}   \nabla c \cdot M \nabla  c \ud x \ud y
 &    + \int_{\pt \Omega_t} \frac{1}{\sqrt{1+|\nabla h|^2}}  [D  e'(c)   \, n_\ell \cdot  \nabla  c +  e(c) v_{\cl} ]  \ud s,
 \end{aligned}
 \end{equation}

From the boundary condition \eqref{3D_bc_c}, we have
\begin{equation}
D  e'(c)    \, n_\ell \cdot  \nabla  c +  e(c) v_{\cl} = - e'(c)  c \,  v_{\cl} +  e(c) v_{\cl} = \gamma(c) v_{\cl}.
\end{equation}
Thus, this, together with \eqref{D_energy2}, implies  the rate of change of the surface energy $\F_0$
\begin{equation}\label{3Denergy_D}
\begin{aligned}
 &\frac{\ud}{\ud t}\int_{\Omega_t} e(c) \sqrt{1+|\nabla h|^2} \ud x \ud y \\
 =& \int_{\Omega_t} \gamma(c)  h_t H -D \frac{e''(c) }{\sqrt{1+|\nabla h|^2}}   \nabla c \cdot M \nabla  c \ud x \ud y  + \int_{\pt \Omega_t} \gamma(c)\cos \theta_{\cl} \,  v_{\cl}  \ud s,
 \end{aligned}
\end{equation}
where we used $\cos \theta_{\cl}  = \frac{1}{\sqrt{1+|\nabla h|^2}}$ on $\pt \Omega_t$.
 
 From \eqref{3Denergy_D} and the rate of change of the surface energy for the bottom part
$$
  \frac{\ud}{\ud t} (\gamma_{\nsl}-\gamma_{\nsg}) \int_{\Omega_t} \ud x \ud y =  (\gamma_{\nsl}-\gamma_{\nsg}) \int_{\pt \Omega_t}v_{\cl} \ud s,
$$
we finally obtain the rate of change of  the total surface energy
\begin{equation}\label{SurfaceE}
\begin{aligned}
 &\frac{\ud}{\ud t}\bbs{\int_{\Omega_t} e(c) \sqrt{1+|\nabla h|^2} \ud x \ud y + (\gamma_{\nsl}-\gamma_{\nsg}) \int_{\Omega_t} \ud x \ud y} \\
 =& \int_{\Omega_t} \gamma(c)  h_t H -D \frac{e''(c) }{\sqrt{1+|\nabla h|^2}}   \nabla c \cdot M \nabla  c \ud x \ud y  + \int_{\pt \Omega_t} \bbs{ \gamma(c)\cos \theta_{\cl} +   \gamma_{\nsl}-\gamma_{\nsg}}   v_{\cl}  \ud s.
 \end{aligned}
\end{equation}

With the volume constraint $V$, we take the total free energy  of the droplet  as
\begin{equation}\label{3Denergy}
\F(h(t), \Omega_t , \lambda(t))= \int_{\Omega_t}  e(c) \sqrt{1+ |\nabla u|^2} \ud x \ud y + (\gamma_{\nsl}-\gamma_{\nsg}) \int_{\Omega_t} \ud x \ud y - \lambda(t) \bbs{ \int_{\Omega_t} h \ud x \ud y - V} ,
\end{equation}
where $\lambda(t)
$ is a Lagrangian multiplier.
Thus, given $h(x,y,t)$ and $c(c,y,t)$, the rate of change of the total free energy can be regarded as a functional of $h_t, v_{\cl}$. Denote
\begin{equation}\label{energyD_sec2}
\begin{aligned}
 \frac{\ud}{\ud t} \F (h_t, v_{\cl}; h, c) := \frac{\ud}{\ud t} \F  
    =& -  \int_{\Omega_t}  \bbs{ -\gamma(c) H + \lambda } h_t \ud x \ud y
 - D \int_{\Omega_t}   \frac{e''(c) }{\sqrt{1+|\nabla h|^2}} \nabla c \cdot    M \nabla  c \ud x \ud y \\
    &+ \int_{\pt \Omega_t}\bbs{\gamma(c) \cos \theta_{\cl}  + (\gamma_{\nsl}-\gamma_{\nsg}) }v_{\cl} \ud s.
\end{aligned}
\end{equation}
The derivation with some additional potential forces is standard and  will not be included  here.

\subsection{The Onsager principle and the governing equations in 3D}\label{sec2.2_E}
{\blue In Section \ref{sec2.1_K}, we used the normal velocity $v_n$ and the contact line speed $v_{\cl}$ to give kinematic descriptions including (i) the motion of the capilary surface \eqref{velocityNN}, (ii) the continuity equation of the surfactant \eqref{tm_mul} and (iii) the rate of change of total energy \eqref{energyD_sec2}. Now we determine these velocities $v_n, v_{\cl}$ from energetic considerations via Onsager's principle.}

From energetic considerations, we choose the following Rayleigh dissipation functional
\begin{equation}\label{RayQ}
Q(h_t, v_{\cl}; h, c):= \frac{\beta}2\int_{\Omega_t} \frac{ h_t^2}{\sqrt{1+|\nabla h |^2}} \ud x \ud y + \frac{\ssr }{2} \int_{\pt \Omega_t} |v_{\cl}|^2 \ud s +\frac{D}{2} \int_{\Omega_t}   \frac{e''(c) }{\sqrt{1+|\nabla h|^2}} \nabla c \cdot    M \nabla  c \ud x \ud y,
\end{equation}
 where $\beta$ represents the friction coefficient for the normal motion of the capillary surface, $\ssr $ represents the friction coefficient for the moving contact lines and the last term represents the dissipation due to the diffusion of the surfactant. We will see the first term in \eqref{RayQ} leads to the motion by mean curvature of the capillary surface \cite{Grunewald_Kim_2011}.
Then minimizing the Rayleighian \cite{Doi_2013} 
\begin{equation}\label{ray}
\Ray (h_t,v_{\cl}; h, c):= Q(h_t, v_{\cl}; h, c)+   \frac{\ud}{\ud t} \F (h_t, v_{\cl}; h, c) 
\end{equation}
with respect to $(h_t, v_{\cl})$ gives the governing equation
\begin{equation}\label{mmc}
\begin{aligned}
 & \frac{\beta}{\sqrt{1+|\nabla h|^2}} h_t=- \gamma(c) H + \lambda,\\
&  \ssr  v_{\cl}=-\gamma(c)\cos \theta_{\cl}  - (\gamma_{\nsl}-\gamma_{\nsg}),
\end{aligned}
\end{equation}
where the right hand side $F_s = -\gamma(c)\cos \theta_{\cl}  - (\gamma_{\nsl}-\gamma_{\nsg})$ is exactly the surfactant-dependent unbalanced Young force.

Combining this with \eqref{3D_ceq}, the full system is
\begin{equation}\label{3Dfull}
\left\{
\begin{aligned}
 & \frac{\beta}{\sqrt{1+|\nabla h|^2}} h_t= -\gamma(c) H + \lambda, \quad  h(x,y)|_{\pt \Omega_t} =0,\\
 &  c_t - \frac{h_t}{\sqrt{1+|\nabla h|^2}} \nabla \cdot \bbs{\frac{c \nabla h}{\sqrt{1+|\nabla h|^2}}}=D \Delta_s c, \quad   c v_{\cl} + D \,  n_\ell  \cdot \nabla c  \big|_{\pt \Omega_t} =0,\\
&  \ssr  v_{\cl}=-\gamma(c)\cos \theta_{\cl}   - (\gamma_{\nsl}-\gamma_{\nsg}), \quad \text{ on } \pt \Omega_t,\\
 & \int_{\Omega_t} h \ud x \ud y = V.
\end{aligned}
\right.
\end{equation}
This system can be regarded as (i) the linear response relation of $v_n$ to the Laplace pressure $\gamma(c) H$, (ii) the linear response relation of $v_{\cl}$ to the   surfactant-dependent unbalanced Young force $F_s$, and (iii) the transport of the insoluble surfactant on the capillary surface.

As a consequence, the energy dissipation relation is
\begin{equation}\label{ED}
\begin{aligned}
  \frac{\ud}{\ud t} \F  
    =& -  \int_{\Omega_t}  \frac{\beta h_t^2}{\sqrt{1+|\nabla h|^2}} \ud x \ud y
 - D \int_{\Omega_t}   \frac{e''(c) }{\sqrt{1+|\nabla h|^2}}  \nabla c \cdot   M \nabla  c \ud x \ud y - \int_{\pt \Omega_t}\ssr |v_{\cl}|^2 \ud s
 \\
 =& -2 Q.
   \end{aligned}
\end{equation}
In physics, $\frac{2Q}{T}$ is denoted as $\dot{S}$, the entropy production rate. 

We point out that the first dissipation term in Rayleigh dissipation functional \eqref{RayQ} is not a standard one.  Instead, the standard dissipation functional includes the dissipation due to the viscosity of fluids inside the droplet, for which we will also give a simple derivation using Onsager's principle in Section \ref{sec_Mar}. However, our choice of the Rayleigh dissipation functional \eqref{RayQ} allows us to 
 study the purely geometric motion of the droplets and has the following advantages.  (i) For small droplets, it captures the essential physics and  the leading behaviors of the droplet dynamics; (ii) it  satisfies the Onsager principle in physics and thus has a gradient flow structure so that this simplified model is  friendly for theoretical studies; (iii) this model is also computationally efficient because it does not need to compute the fluids inside the droplets and the numerical schemes can be easily adapt to more complicated physical examples such as inclined textured substrates, the electrowetting and the surfactant dynamics considered here. 


\section{Surfactant-induced Marangoni stress and viscous flow}\label{sec_Mar}
In this section, we derive the surfactant-induced Marangoni flow for  droplets on a substrate by including the viscous bulk fluids inside the droplets. 
Including the fluid viscosity dissipation in the Rayleigh dissipation functional, instead of the  first term in \eqref{RayQ}, will lead to the  Stokes equations for fluids inside droplets  \cite{Karapetsas_Craster_Matar_2011, Zhang_Xu_Ren_2014, Karapetsas_Sahu_Matar_2016} or the  thin film equation in the lubrication approximation \cite{Karapetsas_Sahu_Matar_2016, xu2016variational}. Although different forms of viscous flow models coupled with moving contact lines and surfactant transport were derived previously, we  adapt the energy law on the capillary surface \eqref{CL_energy_ex} to the general case, i.e., the surfactant moves on the capillary surface along with both the normal velocity $v_n$ and tangential velocity $v_s$ of the fluids,  and then use Onsager's principle to give a simple derivation for the viscous bulk fluids inside the droplet coupled with the Marangoni flow induced by the surfactant. {\blue We will see in \eqref{M_stress} and \eqref{eq_Mar} that the variation of the total surface energy with a surfactant-dependent surface tension will exerts an additional force $\nabla_s \gamma(c)$
for the bulk fluids at the capillary surface $S_t$. This surface gradient of the effective surface tension induces a Marangoni flow, and thus this phenomena is called Marangoni effect.}
At the end of this section, we also point out the two cases with or without bulk fluids inside the droplet are indeed quite similar in terms of the linear response relation $u=\mathcal{K}F$; see \eqref{linearRes}.
In Proposition \ref{prop_korn}, we will prove the Rayleigh dissipation functional for the viscous flow case is indeed stronger than
the one for the motion by mean curvature of the capillary surface.
 However,  the bulk fluids cases with {\blue both the hydrodynamic effect of the viscous bulk fluids inside the droplet and the surfactant effect on the moving surface,   i.e., the case that the dynamics of the bulk fluids is described by the Stokes equation  coupled with the advection-diffusion of the surfactant on the moving capillary surface (see \eqref{eq_Mar}),} require additional computations for bulk fluids. Thus   the  computational strategies presented in Section \ref{sec_num} shall be modified and will be left  as a  future research.

{\blue Now let us first state the idea of  derivations for the governing equations \eqref{eq_Mar} for the general ``free-slip'' case, i.e., the surfactant moves on the capillary surface along with both the normal velocity $v_n$ and tangential velocity $v_s$ of the fluids inside the droplet. For this general case,  given an underlying velocity $v$,  the kinematic description for the capillary surface is still \eqref{eve}, i.e., $h_t + v_1 h_x + v_2 h_y = v_3$, and the continuity equation for the surfactant becomes \eqref{ceq_conven}.
Below,  we will follow the procedures in Section \ref{sec_2} to first give the kinematic descriptions for the rate of change of total energy in Section \ref{Sec3_k}. After the calculations for the rate of change of total energy $\dot{\F}$ in \eqref{final_rate}, it quantifies the work done by the open system (capillary surface laid by surfactant and its contact line) against friction \cite{goldstein2002classical}. Then we determine the velocity fields, including fluid velocity $v$ and contact line speed $v_{\cl}$, via energy considerations in Section \ref{sec3_EE}, i.e.,  via Onsager's principle by  introducing a new Rayleigh dissipation functional $Q$. These immediately yield the governing equations and the energy dissipation law for the bulk fluids coupled with transport of surfactant on the evolutionary surface; see \eqref{eq_Mar}. In the end of this section, we give  a proof of Onsager's reciprocal relation and obtain a lower bound estimate for energy dissipation in Proposition \ref{prop_korn}, which helps us to characterize the steady solution as a spherical cap laid by constant-concentrated surfactant.

\subsection{Kinematic descriptions for the rate of change of total energy}\label{Sec3_k}
In the following subsections,  we will adapt the same kinematic descriptions as in Section \ref{sec_2} to the general ``free-slip'' case for surfactant with $v_s$, derive the corresponding boundary conditions of the surfactant at the contact line and compute the rate of change of the total surface  energy. With the convection contribution, transport equation \eqref{3D_c_con} for $c$ becomes 
\begin{equation}\label{ceq_conven}
 \pt_t \bbs{c {\sqrt{1+|\nabla h|^2}} } - \nabla \cdot \bbs{\frac{ c h_t}{\sqrt{1+|\nabla h|^2}} \nabla h}  +  \nabla \cdot \bbs{c \sqrt{1+|\nabla h|^2} \left( \begin{array}{c}
f\\
g
\end{array} \right) }  =  D  \nabla \cdot  \bbs{\frac{1}{\sqrt{1+|\nabla h|^2}}  M \nabla  c}.
\end{equation}
We impose a non-flux boundary condition \begin{equation}\label{bc-noflux}
n_{\ell} \cdot \nabla c \big|_{\pt \Omega_t}=0
\end{equation} 
 for $c$ to preserve the total mass of the surfactant. }
\subsubsection{Conservation of the total mass of the surfactant}
Using \eqref{ceq_conven} and the same calculus derivations as  \eqref{tm223} except for adding $(f,g)$, we have
\begin{equation}\label{final1}
\begin{aligned}
&\frac{\ud }{\ud t} \int_{\Omega_t} c \sqrt{1+|\nabla h|^2} \ud x \ud y\\
=& \int_{\pt \Omega_t} \frac{D}{\sqrt{1+|\nabla h|^2}} n_\ell \cdot \nabla c + \int_{\pt \Omega_t} c\frac{1}{\sqrt{1+|\nabla h|^2}} v_{\cl} - c \sqrt{1+|\nabla h|^2} \left( \begin{array}{c}
f\\
g
\end{array} \right) \cdot n_{\ell}  \ud s.
\end{aligned}
\end{equation}
The first term on the right-hand-side vanishes due to the no-flux boundary condition \eqref{bc-noflux}. 
Recall \eqref{velocity}.
Suppose we impose the non-penetration boundary condition for the bulk fluid velocity $v \cdot e_z|_{\pt \Omega_t} = 0$. Then
\begin{equation}
v \cdot e_z = \frac{v_n}{\sqrt{1+|\nabla h|^2}} + h_x f + h_y g   = 0 \quad  \text{ on }\pt \Omega_t. 
\end{equation}
Then by $n_\ell= -\frac{\nabla h}{|\nabla h|}$, we have
\begin{equation}
\frac{v_n}{\sqrt{1+|\nabla h|^2}} = |\nabla h|\left( \begin{array}{c}
f\\
g
\end{array} \right) \cdot n_{\ell}.
\end{equation}
From this and the continuity condition \eqref{com_22}, we know  the last two terms in \eqref{final1} also vanish
\begin{equation}\label{temp255}
\frac{1}{\sqrt{1+|\nabla h|^2}} v_{\cl} -  \sqrt{1+|\nabla h|^2} \left( \begin{array}{c}
f\\
g
\end{array} \right) \cdot n_{\ell} =0.
\end{equation}
Hence \eqref{final1} yields $\frac{\ud }{\ud t} \int_{\Omega_t} c \sqrt{1+|\nabla h|^2} \ud x \ud y=0.$

\subsubsection{Calculations of Marangoni stress  induced by the tangential convection of surfactant}
Recall $(v_1, v_2)(x,y,t)$ is the $xy$-component of the velocity of the moving capillary surface. Then \eqref{ceq_conven} is recast as
\begin{equation}\label{tm_c31_con}
 \pt_t \bbs{c {\sqrt{1+|\nabla h|^2}} }  +  \nabla \cdot \bbs{c \sqrt{1+|\nabla h|^2}  \left( \begin{array}{c}
v_1\\
v_2
\end{array} \right) }  = D  \nabla \cdot  \bbs{\frac{1}{\sqrt{1+|\nabla h|^2}}  M \nabla  c}.
\end{equation}
Then by same calculations as \eqref{CL_energy_ex}, we obtain the change of the surface energy 
\begin{equation}\label{energy_law_33}
\begin{aligned}
 &\pt_t \bbs{e(c){\sqrt{1+|\nabla h|^2}} }  +  \nabla \cdot \bbs{e(c) \sqrt{1+|\nabla h|^2}  \left( \begin{array}{c}
v_1\\
v_2
\end{array} \right)}\\
   =&  \gamma(c)
 \bbs{ \pt_t {\sqrt{1+|\nabla h|^2}}  +  \nabla \cdot \bbs{ \sqrt{1+|\nabla h|^2}  \left( \begin{array}{c}
v_1\\
v_2
\end{array} \right) } } + D \sqrt{1+|\nabla h|^2}  e'(c) \Delta_s  c\\
 =:& \gamma(c)I + D \sqrt{1+|\nabla h|^2}  e'(c) \Delta_s  c.
 \end{aligned}
\end{equation}
Here using the identity \eqref{MC_ID}, $\gamma(c) I$ can be further simplified as
\begin{align}
\gamma(c)I =&   \gamma(c) h_t H 
- \nabla \gamma(c)  \cdot \bbs{  \sqrt{1+|\nabla h|^2} \left( \begin{array}{c} f\\g\end{array} \right) }
+  \nabla \cdot \bbs{  \gamma(c) \sqrt{1+|\nabla h|^2} \left( \begin{array}{c} f\\g\end{array} \right) }\nonumber\\
=&  - \sqrt{1+|\nabla h|^2}  \bbs{-
 \gamma(c) v_n H + \left( \begin{array}{c} f\\g\end{array} \right) \cdot \nabla \gamma(c) }
+  \nabla \cdot \bbs{  \gamma(c) \sqrt{1+|\nabla h|^2} \left( \begin{array}{c} f\\g\end{array} \right) }. \label{tmgamma}
\end{align}

{\blue Now we simplify the first term in $\gamma(c) I$ to see that besides the Laplace pressure, there is an additional force, brought by the surface gradient of $\gamma(\cc)$.
Recall 
$v = v_n n + v_s.$
 Denote 
\begin{equation}\label{def_FF}
F := -\gamma(c) H n + \nabla_s \gamma(\cc).
\end{equation}
Here the additional force term, surface gradient $\nabla_s \gamma(\cc)$, is called the Marangoni stress,  so the total  capillary force density $F$ exerting to the environment consists of both the negative Laplace pressure done  by the capillary surface  and the Marangoni stress caused by the surface gradient of  surfactant-dependent surface tension.}
We now prove the first term in $\gamma(c) I$ satisfies 
\begin{equation}\label{claimG}
 -\gamma(c) v_n H + \left( \begin{array}{c} f\\g\end{array} \right) \cdot \nabla_{xy} \gamma(c) = v\cdot F.
\end{equation}
{\blue In other words, the first term in \eqref{tmgamma} can be recast as the work done by the capillary surface per unit time $v\cdot F$. }
Indeed, from the orthogonality and definition of $F$ in \eqref{def_FF}, we have
\begin{equation}\label{M_stress}
v \cdot F = -v_n \gamma(c) H + v_s \cdot \nabla_s \gamma(\cc)\quad  \text{ on } z=h(x,y,t).
\end{equation}
By the definition of $v_s$ and surface gradient, the last term is
\begin{align*}
v_s \cdot \nabla_s \gamma(\cc) =& (f \tau_1 + g \tau_2) \cdot [(I-n\otimes n) \nabla \gamma(\cc)] \\
=& (f \tau_1 + g \tau_2) \cdot \nabla \gamma(\cc) =   \left( \begin{array}{c} f\\g\end{array} \right) \cdot \nabla_{xy} \gamma(c).
\end{align*}

\subsubsection{Rate of change of total surface energy}

Recall the calculations for energy change  \eqref{energy_law_33}  and its relation with the Marangoni stress $F$ in \eqref{claimG}.  Using the same calculations as \eqref{D_energy11}, we can simplify the second  term in  \eqref{tmgamma}
 \begin{equation}
\begin{aligned}
 &\frac{\ud}{\ud t}\int_{\Omega_t} e(c) \sqrt{1+|\nabla h|^2} \ud x \ud y \\
 =& - \int_{\Omega_t} v \cdot F \sqrt{1+|\nabla h|^2}  \ud x \ud y   - \int_{\Omega_t}D \frac{e''(c) }{\sqrt{1+|\nabla h|^2}}   \nabla c \cdot M \nabla  c \ud x \ud y\\
 &+  \int_{\pt \Omega_t} e(c)\frac{ h_{t}}{\sqrt{1+|\nabla h|^2}} \nabla h \cdot n_\ell  - c e'(c) \sqrt{1+|\nabla h|^2}  \left( \begin{array}{c} f\\g\end{array} \right) \cdot n_\ell     \ud s 
 + \int_{\pt \Omega_t}e(c) \sqrt{1+|\nabla h|^2} v_{\cl} \ud s,
 \end{aligned}
 \end{equation}
 where we used the no-flux boundary condition \eqref{bc-noflux}.
 Using \eqref{temp255} and \eqref{useID1}, 
  \begin{equation}
\begin{aligned}
 &\frac{\ud}{\ud t}\int_{\Omega_t} e(c) \sqrt{1+|\nabla h|^2} \ud x \ud y \\
 =& - \int_{\Omega_t} v \cdot F \sqrt{1+|\nabla h|^2}  \ud x \ud y   - \int_{\Omega_t}D \frac{e''(c) }{\sqrt{1+|\nabla h|^2}}   \nabla c \cdot M \nabla  c \ud x \ud y\\
 &+  \int_{\pt \Omega_t} -e(c)\frac{ |\nabla h|^2}{\sqrt{1+|\nabla h|^2}} v_{\cl}  - c e'(c) \frac{1}{\sqrt{1+|\nabla h|^2} } v_{\cl}     \ud s 
 + \int_{\pt \Omega_t}e(c) \sqrt{1+|\nabla h|^2} v_{\cl} \ud s\\=& - \int_{\Omega_t} v \cdot F \sqrt{1+|\nabla h|^2}  \ud x \ud y   - \int_{\Omega_t}D \frac{e''(c) }{\sqrt{1+|\nabla h|^2}}   \nabla c \cdot M \nabla  c \ud x \ud y +  \int_{\pt \Omega_t}\gamma(c) v_{\cl} \frac{1}{\sqrt{1+|\nabla h|^2} }  \ud s. 
 \end{aligned}
 \end{equation}
 
 Therefore in the presence of convention contribution of surfactant, the rate of change of the total surface energy becomes
 \begin{equation}\label{SurfaceE_con}
\begin{aligned}
 &\frac{\ud}{\ud t} \Big[ \int_{\Omega_t} e(c) \sqrt{1+|\nabla h|^2} \ud x \ud y + (\gamma_{\nsl}-\gamma_{\nsg}) \int_{\Omega_t} \ud x \ud y \Big] \\
 =& - \int_{\Omega_t} v \cdot F \sqrt{1+|\nabla h|^2}  \ud x \ud y   - \int_{\Omega_t}D \frac{e''(c) }{\sqrt{1+|\nabla h|^2}}   \nabla c \cdot M \nabla  c \ud x \ud y - \int_{\pt \Omega_t} F_s  v_{\cl}  \ud s,
 \end{aligned}
\end{equation}
where $F_s =  -\gamma(c)\cos \theta_{\cl} +  \gamma_{\nsg} - \gamma_{\nsl}$ is the surfactant-dependent unbalanced Young force.

\subsection{ Energetic descriptions: Stokes flow for bulk fluids and governing equations derived by Onsager's principle}\label{sec3_EE}

{\blue As we used the fluid velocity $v$ and the contact line speed $v_{\cl}$ to give kinematic descriptions including  motion of the capilary surface, the continuity equation of the surfactant   and the rate of change of total energy \eqref{SurfaceE_con}, it remains to   determine these velocities $v, v_{\cl}$ from energetic considerations.}
Using Onsager's principle and a new Rayleigh dissipation functional, we give   derivations for the governing equations of the surfactant induced Marangoni flow  coupled with moving contact lines.

{\blue Let $u(x,y,z,t)$ be the velocity of the bulk fluids.}
 Assume the velocity $v$ of the capillary surface coincides with the velocity $u$ of the bulk fluids  {\blue restricted on the capillary surface, i.e., $v(x,y,t)=u(x,y,z,t)\big|_{z=h(x,y,t)}$}.

First, we  impose the non-penetration boundary condition for the bottom of the
droplet
\begin{equation}\label{non-p}
u \cdot n = 0  \quad \text{ on } \Omega_t
\end{equation}
and consider an incompressible fluid satisfying $\nabla \cdot u = 0$ inside the droplet $A_t$.

Given $h(x,y,t)$ and $c(x,y,t)$, then
the rate of change of the total surface energy  \eqref{SurfaceE_con} can be regarded as a
linear functional of $u, v_{\cl}$ and we denote it as
 \begin{equation}\label{final_rate}
\begin{aligned}
\dot{\F}(u, v_{\cl}; h, c) := - \int_{\Omega_t} u \cdot F \sqrt{1+|\nabla h|^2}  \ud x \ud y   - \int_{\Omega_t}D \frac{e''(c) }{\sqrt{1+|\nabla h|^2}}   \nabla c \cdot M \nabla  c \ud x \ud y - \int_{\pt \Omega_t} F_s  v_{\cl}  \ud s.
 \end{aligned}
\end{equation}
{\blue We remark 
 the rate of change of total energy $\dot{\F}$ in \eqref{final_rate} quantifies the work done by the open system (capillary surface laid by surfactant and its contact line) against friction \cite{goldstein2002classical}. Then all we need to do is to determine the the fluid velocity $u$ and the contact line speed $v_{\cl}$ via  Onsager's principle by  introducing a proper Rayleigh dissipation functional $Q$.
In general, the choice of the Rayleigh dissipation functional in Onsager's principle is just from phenomenological modeling instead of from physical principle. Therefore, Onsager's principle is valid only for certain class of problems. However, many specific problems in soft matter including diffusion, viscous fluids and surfactant belong to this class. The resulting governing equations determined by the choice of Rayleighian  are consistent with some well-accepted or experimentally tested models, for instance, the Navier-Stokes equations and the stokes equations in the current section; see more worked out examples in \textsc{Doi}'s book \cite{doi2013soft}.}

Second, given $h(x,y,t)$ and $c(x,y,t)$, introduce 
the Rayleigh dissipation functional $Q$
\begin{equation}\label{Ray_stokes}
\begin{aligned}
Q(u,v_{\cl}; h,c):=  \frac{\mu}{4} \int_{A_t} (\nabla u + \nabla u^{\top}) : (\nabla u + \nabla u^{\top})
 \ud V
 &+\frac{D}2 \int_{\Omega_t}   \frac{e''(c) }{\sqrt{1+|\nabla h|^2}}  \nabla c \cdot   M \nabla  c \ud x \ud y\\ 
 &+ \frac{\mu}{2\alpha}\int_{\Omega_t}  u^2 \ud x \ud y  
 + \frac{\ssr  }2 \int_{\pt \Omega_t}|v_{\cl}|^2 \ud s,
 \end{aligned}
\end{equation}
where $\mu$ is the dynamics viscosity for the bulk fluids and $\alpha$ is the slip length.

Given $h(x,y,t)$ and $c(x,y,t)$, define Rayleighnian as
\begin{equation}
\Ray(u, v_{\cl}; h,c) := Q(u, v_{\cl}; h, c) + \dot{\F}(u, v_{\cl}; h, c).
\end{equation}
Then based on Onsager's principle,  we minimize $\Ray(u, v_{\cl}; h,c)$ w.r.t the velocity $u, v_{\cl}$. This  yields the following governing equations, whose derivations are given in three steps below.

After incorporating the transport  equation \eqref{ceq_conven} for surfactant $c$ and the no-flux boundary condition \eqref{bc-noflux}, the minimization of $\Ray(u, v_{\cl}; h,c)$  gives the governing equations
\begin{equation}\label{eq_Mar}
\begin{aligned}
&\left\{ \begin{array}{cc}
  \nabla p=\mu \Delta u  & \text{ in } A_t,\\
\nabla \cdot u = 0 &  \text{ in } A_t, \\
\sigma n = F & \text{ on } S_t,\\
  \frac{\alpha}{\mu} \tau \cdot \sigma n + \tau \cdot u = 0, \quad u \cdot n = 0 & \text{ on } \Omega_t,
\end{array}
\right. \\
&\left\{ \begin{array}{cc}
\pt_t h +  u_1 \pt_x h +  u_2 \pt_y h = u_3 & \text{ on } \Omega_t, \\
h = 0 & \text{ on }\pt\Omega_t\\
\ssr  v_{\cl} = F_s & \text{ on }\pt \Omega_t,
\end{array}
\right. \\
&\left\{ \begin{array}{cc}
\pt_t  \bbs{c \sqrt{1+ |\nabla_{xy} h|^2}} + \nabla_{xy} \cdot \bbs{c \sqrt{1+|\nabla_{xy} h|^2} \left( \begin{array}{c}
u_1\\
u_2
\end{array} \right) }=D  \nabla_{xy} \cdot  \bbs{\frac{1}{\sqrt{1+|\nabla_{xy} h|^2}}  M \nabla_{xy} c}  & \text{ on } \Omega_t,\\
\nabla_{xy} c \cdot  n_{\ell} = 0 & \text{on } \pt \Omega_t,
\end{array}
\right.
\end{aligned}
\end{equation}
where $u_1(x,y,t) =u_x(x,y,h(x,y,t),t),\ u_2(x,y,t) =u_y(x,y,h(x,y,t),t)$, $u_3(x,y,t) =u_z(x,y,h(x,y,t),t)$ and
\begin{equation}\label{notation1}
\sigma = -pI + \mu (\nabla u + \nabla u^\top), \quad F = -\gamma(c) H n + \nabla_s \gamma(\cc),  \quad F_s =  -\gamma(c)\cos \theta_{\cl} +  \gamma_{\nsg} - \gamma_{\nsl}.
\end{equation}
The first group of \eqref{eq_Mar} is the stationary Stokes equation inside the droplet coupled with the traction boundary condition balanced with the total capillary force density $F$ and Navier slip boundary condition at the bottom. The second group of \eqref{eq_Mar} is the evolution of the capillary surface induced by the fluid velocity and the moving contact line as a linear response to the surfactant-dependent unbalanced Young force $F_s$. The third group of \eqref{eq_Mar} is the transport equation for the insoluble surfactant with the no-flux boundary condition on the contact line $\pt \Omega_t.$ We particularly point out that the Dirichlet boundary condition $h=0$ on $\pt \Omega_t$ is necessary \cite{marchand2011surface} because  ``the projection of the capillary
forces onto the vertical axis is balanced out by a force of
reaction exerted by the solid'' \cite[p.18]{de2013capillarity}. We refer to \cite{Tice21} for the wellposedness of this model for the 2D single droplet without surfactant.
We remark in the case $\alpha\to0$, the Navier slip boundary condition in the first group of \eqref{eq_Mar} is reduced to the nonslip boundary condition $u=0$ on $\Omega_t.$ In this case, the Rayleigh dissipation functional becomes
\begin{equation}\label{Ray_stokes_n}
\begin{aligned}
Q=  \frac{\mu}{4} \int_{A_t} (\nabla u + \nabla u^{\top}) : (\nabla u + \nabla u^{\top})
 \ud V
 +\frac{D}2 \int_{\Omega_t}   \frac{e''(c) }{\sqrt{1+|\nabla h|^2}}  \nabla c \cdot   M \nabla  c \ud x \ud y 
 + \frac{\ssr  }2 \int_{\pt \Omega_t}|v_{\cl}|^2 \ud s.
 \end{aligned}
\end{equation}

The derivations of the governing equations \eqref{eq_Mar} can be summarized as the following three steps.

Step 1. To impose the incompressible condition, introduce the Lagrangian multiplier $p$. Then we take the first  variation of $\Ray$ with perturbations $u+\eps \tilde{u}, p+ \eps \tilde{p}$ such that $(\tilde{u}, \tilde{p})$ are compact supported in the open set $A_t$,
\begin{equation}
0 = \frac{\ud}{\ud \eps}\Big|_{\eps=0} \bbs{\Ray(u+\eps \tilde{u}, v_{\cl}; h,c) -  \int_{A_t} (p+ \eps \tilde{p}) \nabla \cdot (u+ \tilde{u}) }.
\end{equation}
This implies 
the static Stokes equation inside $A_t$ 
\begin{equation}\label{St-eq}
\begin{aligned}
\nabla \cdot \sigma = 0, \quad \nabla \cdot u=0,\\
\sigma = -pI + \mu (\nabla u + \nabla u^\top).
\end{aligned}
\end{equation}

Step 2. 
From the non-penetration boundary condition \eqref{non-p}, we take the first  variation of $\Ray$ with perturbations $u+\eps \tilde{u}$ satisfying
$ \tilde{u} \cdot n \big|_{z= 0} =0$. Using \eqref{St-eq},
\begin{equation}
\begin{aligned}
0 =& \frac{\ud}{\ud \eps}\Big|_{\eps=0} \Ray(u+\eps \tilde{u}, v_{\cl}; h,c)\\
=& \frac{\mu}{2} \int_{A_t} (\nabla u + \nabla u^{\top}) : (\nabla \tilde{u} + \nabla \tilde{u}^{\top})
 \ud V + \frac{\mu}{\alpha}\int_{\Omega_t}  u \cdot \tilde{u} \ud x \ud y  - \int_{S_t} \tilde{u} \cdot F \ud s   \\
 =& \int_{A_t} \sigma : \nabla \tilde{u} \ud V  + \frac{\mu}{\alpha}\int_{\Omega_t}  u \cdot \tilde{u} \ud x \ud y  - \int_{S_t} \tilde{u} \cdot F \ud s\\
  =& \int_{\Omega_t} \tilde{u} \cdot (\sigma n + \frac{\mu}{\alpha} u) \ud x \ud y + \int_{S_t} \tilde{u} \cdot \bbs{\sigma n - F } \ud s . 
\end{aligned}
\end{equation}
By the arbitrary of $\tilde{u}$, this implies the Navier slip boundary condition for the bottom part
\begin{equation}
\frac{\alpha}{\mu} \tau \cdot \sigma n + \tau \cdot u =0 \quad \text{ on } \Omega_t
\end{equation}
and the traction boundary condition on the capillary surface
\begin{equation}
\sigma n = F \quad \text{ on } S_t.
\end{equation}

Step 3. Taking first variation of $R$ w.r.t $v_{\cl}+ \eps \tilde{v}_{\cl}$, we obtain the contact line speed
\begin{equation}
\ssr  v_{\cl} = F_s \quad \text{ on } \pt \Omega_t.
\end{equation}
Thus after incorporating the transport of the surfactant,  we obtain the governing equations \eqref{eq_Mar} of the surfactant induced Marangoni flow for droplets on a substrate.
As a consequence, the energy dissipation law is
\begin{equation}\label{E_dis}
\dot{\F} = -2Q.
\end{equation}

Now we explain Onsager's reciprocal relation for the linear response of $u, v_{\cl}$ to the two unbalanced forces $F, F_s$. Given $c(x,y,t)$ and $h(x,y,t)$, define an operator
$$\mathcal{K}: L^2(S_t; \bR^3) \to L^2(S_t;\bR^3), \quad  F \mapsto u|_{S_t}$$ 
such that $u$ solves the Stokes equations with traction force $F$ (first group in \eqref{eq_Mar}). 
Then 
 the velocity fields $u, v_{\cl}$ solved from the governing equation \eqref{eq_Mar} satisfy the following linear response relations
\begin{equation}\label{linearRes}
\begin{array}{cc}
u = \mathcal{K} F \quad &\text{ on }S_t,\\
v_{\cl} = \frac{1}{\ssr} F_s \quad & \text{ on }\pt \Omega_t.
\end{array}
\end{equation}

\begin{prop} \label{prop_korn}
 $\mathcal{K}$ is a  bijective and self-adjoint   operator from $L^2(S_t; \bR^3)$ onto $L^2(S_t;\bR^3)$. $\mathcal{K}$ is a positive  operator in $L^2(S_t;\bR^3)$. Furthermore, for the no-slip boundary condition case, i.e., $\alpha=0$, then $\mathcal{K}$ is a bounded operator in $L^2(S_t;\bR^3)$  satisfying
 \begin{equation}\label{linearCom}
 \int_{S_t} F \cdot \mathcal{K}F \ud s \geq C \int_{S_t} |\mathcal{K}F|^2 \ud s \quad \text{ for any }F \in L^2(S_t;\bR^3). 
 \end{equation}
 \end{prop}
\begin{proof}
First, given any $f\in L^2(S_t;\bR^3)$, the solution to 
\begin{equation}
\left\{ \begin{array}{cc}
  \nabla p=\mu \Delta u  & \text{ in } A_t,\\
\nabla \cdot u = 0 &  \text{ in } A_t, \\
u = f & \text{ on } S_t,\\
  \frac{\alpha}{\mu} \tau \cdot \sigma n + \tau \cdot u = 0, \quad u \cdot n = 0 & \text{ on } \Omega_t,
\end{array}
\right. 
\end{equation}
exists uniquely. This gives a unique $F=\sigma n$ and thus $\mathcal{K}$ is bijective
operator from $L^2(S_t;\bR^3)$ onto $L^2(S_t;\bR^3)$. 

 Second, for any $F_1, F_2 \in L^2(S_t;\bR^3)$, let $u_1 = \mathcal{K} F_1$ and $u_2 = \mathcal{K} F_2$. Then we have
\begin{equation}\label{sys_S}
\begin{aligned}
\int_{S_t} F_1 \cdot  \mathcal{K} F_2 \ud s =& \int_{S_t} \sigma_1 n \cdot u_2  \ud s \\
=&  \frac{\mu}{2} \int_{A_t} (\nabla u_1 + \nabla u_1 ^\top) : (\nabla u_2 + \nabla u_2^\top) \ud V + \frac{\mu}{\alpha} \int_{\Omega_t} u_1 \cdot u_2 \ud x \ud y =\int_{S_t} \mathcal{K}F_1 \cdot   F_2 \ud s,
\end{aligned}
\end{equation}
which, together with $D(\mathcal{K})= L^2(S_t)$, shows $\mathcal{K}$ is self-adjoint. The symmetry \eqref{sys_S} is known as Lorentz's reciprocal theorem for the Stokes flow.

Third, the positivity of $\mathcal{K}$ is directly from 
\begin{equation}
\int_{S_t} F \cdot \mathcal{K} F \ud s \geq 0.
\end{equation}
This equality holds if and only if $F\equiv 0$ because $u|_{\Omega_t}=0$ implies the Korn's inequality \cite[(3)]{Desvillettes_Villani_2002}
\begin{equation}\label{Korn}
\int_{A_t} |\nabla u|^2 \ud V \leq C \int_{A_t} |\nabla u + \nabla u ^\top|^2 \ud V,
\end{equation}
where $C$ is a generic constant.


Fourth, in the case $\alpha=0$.  Combining the trace theorem, the standard Poincare's inequality and Korn's inequality \eqref{Korn}, we know
\begin{equation}
\int_{S_t} |u|^2 \ud s \leq C \int_{A_t} \bbs{|u^2| + |\nabla u|^2} \ud V \leq C \int_{A_t} |\nabla u + \nabla u ^\top|^2 \ud V .
\end{equation}
This concludes \eqref{linearCom}.
\end{proof} 

{\blue The first statement in Proposition \ref{prop_korn} proves the linear response operator $\mathcal{K}$ is positive self-adjoint in $L^2(S_t;\bR^3)$. The symmetry \eqref{sys_S} of $\mathcal{K}$ is exactly the original statement in Lorentz's reciprocal theorem for the Stokes flow, while we interpret it as a positive self-adjoint operator in a functional space. 

The second advantage of Proposition \ref{prop_korn} is we proved the  lower bound of the inverse operator $\mathcal{K}^{-1}$, which means the dissipation term $Q$ in \eqref{E_dis} and \eqref{Ray_stokes_n} has a lower bound. That gives further the energy estimate
\begin{equation}\label{new_E_dis}
\begin{aligned}
\dot{\F} = -2Q\leq - C \int_{S_t} |u|^2 \ud s
 -D \int_{\Omega_t}   \frac{e''(c) }{\sqrt{1+|\nabla h|^2}}  \nabla c \cdot   M \nabla  c \ud x \ud y 
 - {\ssr  } \int_{\pt \Omega_t}|v_{\cl}|^2 \ud s.
 \end{aligned}
\end{equation}
We remark in the geometric motion case, the first equation in \eqref{3Dfull} is the linear  response relation $u = \mathcal{K}F= \frac{1}{\beta} F.$ Thus Proposition \ref{prop_korn} also tells us
the Rayleigh dissipation functional $Q$ in \eqref{Ray_stokes} for the viscous flow is stronger than the
one defined in \eqref{RayQ} for the geometric motion case.
As a consequence of \eqref{new_E_dis}, we know $\int_0^{+\8}\int_{S_t} |u|^2 \ud s \ud t <+\8$. 

The third advantage of Proposition \ref{prop_korn} is the  characterization for the steady solution to \eqref{eq_Mar}. Rigorous proof for the uniform convergence of the dynamic solutions to its steady state is challenging; see \cite{Guo_Tice_2018} for small data convergence without surfactant. However, if the limiting solution (the capillary surface) remains a smooth surface, i.e., assume no fattening phenomenon, we will show below the steady solution $u\equiv  0$ in $\bar{A_t}$, $c \equiv \text{const}$ on $S_t$ and the steady shape $S_t$ is characterized by a spherical cap with constant mean curvature while the contact angle being Young's angle. Specifically:\\
 (i) For the steady solution to  \eqref{eq_Mar} with $Q=0$, we know $\int_{\Omega_t}   \frac{e''(c) }{\sqrt{1+|\nabla h|^2}}  \nabla c \cdot   M \nabla  c \ud x \ud y=0$ and thus $c \equiv  c^*$ is constant on $S_t$;\\
  (ii) As a consequence of \eqref{new_E_dis}, we know  $Q=0$ implies $\int_{S_t} |u|^2 \ud s=0$ and thus $u=0$ on $S_t$. From $u=0$ on $S_t\cap \Omega_t$, the bulk fluids described by the stokes equation satisfy $u\equiv 0$ in $\bar{A_t}$. This uniquely determines $\sigma=-p I$ with a constant pressure $p=\text{const}$ and thus $F=-p n = -\gamma(c^*)Hn$ due to \eqref{notation1} and $c\equiv c^*$;\\
 (iii) Therefore the characterization problem is reduced to solve a spherical cap profile with constant mean curvature $H=\frac{p}{\gamma(c^*)}=\text{const}$.  Denote $R$ as the radius of the spherical cap. Given volume $V$ and total mass $M_0$ of surfactant in \eqref{const}, we solve unknowns $\{\theta, R, c^*\}$ such that
\begin{align}\label{capS}
\cos \theta = \frac{\gamma_{\nsg} - \gamma_{\nsl}}{\gamma(c^*)},\quad
V= \frac{\pi R^3}{3}(2+\cos \theta)(1-\cos \theta)^2,\quad
\frac{M_0}{c^*} = {2\pi R^2}(1-\cos \theta),
\end{align}
where we used the volume and area formula for a spherical cap with radius $R$. Notice $\gamma(c)$ is strictly decreasing w.r.t. $c$. The solvability of these algebraic equations, the uniqueness of the spherical cap solution and the convergence to this steady solution will be left for a future study.
}

\begin{rem}
We also remark under the non-penetration boundary condition $u \cdot n = 0$ on the bottom of the droplet,    the Navier slip boundary condition  on a textured substrate $w(x,y)$ becomes
\begin{equation}
  \frac{\alpha}{\mu} \tau \cdot \sigma n + \tau \cdot u = 0 \quad  \text{ on } \pt A_t \cap \{z=w\},
\end{equation}
which is equivalent to 
\begin{equation}
\alpha \, (n\cdot\nabla)  ( \tau\cdot u ) +  \tau\cdot u  = \alpha u\cdot ( n\cdot\nabla \tau +  n\cdot\nabla \tau ) \quad \text{ on } \pt A_t \cap \{z=w\}.
\end{equation}
In the case $w=0$, then the Navier slip boundary condition is simplified as
\begin{equation}
 \tau\cdot u = \alpha \partial_z ( \tau\cdot u ) \quad  \text{ on } \Omega_t.
\end{equation} 
\end{rem}

\section{Algorithms based on unconditionally stable explicit boundary updates}\label{sec_num}
In this section, we propose a numerical scheme for the droplet dynamics with the surfactant on the moving capillary surface. These mainly rely on  decoupling the motion of the contact lines,  the motion of the capillary surface, and the dynamics of the surfactant on the surface. Therefore, we will adapt  the 1st/2nd order schemes developed in \cite{gao2020gradient} for the pure geometric motion of single droplets and then incorporate the constantly changed dynamic surface tension $\gamma(c)$ due to the dynamics of the surfactant.

To give a clear presentation,  we describe the numerical scheme  for 2D droplets. For the  3D droplets, the construction of the arbitrary Lagrangian-Eulerian method for the moving grids need to be developed and will be left for a future study.
Before this, we first derive the governing equations for 2D droplets laid by surfactant but placed on an inclined textured substrate.
This is described by the motion of the capillary surface, the moving contact lines and the transport of the  surfactant.

\subsection{Contact line dynamics and surfactant effect  for 2D droplets  on an inclined textured substrate}\label{sec_3}

 Given an inclined textured  solid substrate, we follow the convention for studying droplets on an inclined substrate  and choose the Cartesian coordinate system built on an inclined plane with effective inclined angle $\theta_0$ such that $-\frac{\pi}{2}<\theta_0<\frac{\pi}{2}$, i.e.,  $(\tan \theta_0) x$ is the new $x$-axis we choose; see Fig \ref{fig:ill} (b).
With this Cartesian coordinate system,  the textured substrate is described by a graph function $w(x)$ and the droplet is then described by
\begin{equation}
A_t:= \{(x,y); a(t)\leq x\leq b(t) , w(x)\leq y\leq u(x,t)+w(x) \}. 
\end{equation}
The motion of this droplet is described by the relative height function of the capillary surface $u(x,t)\geq 0$ and the  partially wetting domain $a(t)\leq x \leq b(t)$ with free boundaries $a(t), b(t).$

 With the new Cartesian coordinate system, the substrate $w(x)$ and the total height $h(x,t):= u(x,t)+w(x)$, one can use the same lift-up method in Section \ref{sec_lift} to derive the continuity equation for $c(x,t), x\in(a(t),b(t))$, i.e.
\begin{equation}\label{eq-c2D}
 c_t  - v_n \pt_x \bbs{c \frac{h_x }{ \sqrt{1+h_x^2}}}=D \pt_{ss}c, \quad \pt_s = \frac{1}{\sqrt{1+h_x^2}} \pt_x.
\end{equation}
This is equivalent to
\begin{equation}\label{recast_c2D}
\pt_t c- v_n  c_x  \frac{h_x }{ \sqrt{1+h_x^2}}  +  v_n c H=D \pt_{ss}c, \quad  H = -\pt_x \bbs{\frac{h_x}{\sqrt{1+h_x^2}}}=-\frac{h_{xx}}{(1+h_x^2)^{\frac32}}
\end{equation}

 Notice the compatibility condition \eqref{com_22} on the contact line now changes to
\begin{equation}\label{com-com}
h_t|_b = - \pt_x u |_b b', \quad h_t|_a = - \pt_x u|_a a'
\end{equation}
due to $u(x(t), t) = 0$ at $x=a, b$. We will compute the rate of change of total energy  and derive the governing equations for droplets placed an inclined textured surface  below.

\subsection{The rate of change of the total surface energy with a textured substrate}

First, we compute the rate of change of the surface energy for the capillary surface.

Multiplying \eqref{eq-c2D} by $e'(c)\sqrt{1+h_x^2}$, 
 same derivations as \eqref{CL_energy_ex} gives
\begin{equation}\label{ttt}
\pt_t \bbs{e(c) \sqrt{1+h_x^2} } -  \gamma(c) H h_t - \pt_x \bbs{e(c)\frac{h_t h_x}{\sqrt{1+h_x^2}} } = D e'(c) \pt_{ss}c \, \sqrt{1+h_x^2}. 
\end{equation}
On one hand, integration of the left-hand-side of \eqref{ttt} gives
\begin{align*}
&\int_{a(t)}^{b(t)} \pt_t \bbs{e(c) \sqrt{1+h_x^2} } -  \gamma(c) H h_t  \ud x   - e(c)\frac{h_x h_{t}}{\sqrt{1+h_x^2}} \Big|_a^b \\
=& \frac{\ud}{\ud t}\int_{a(t)}^{b(t)}  e(c) \sqrt{1+h_x^2}   \ud x  +  \int_{a(t)}^{b(t)} -\gamma(c) H h_t  \ud x  - e(c)\frac{h_x h_{t}}{\sqrt{1+h_x^2}} \Big|_a^b  - b' e(c) \sqrt{1+h_x^2}\big|_b + a' e(c) \sqrt{1+h_x^2}\big|_a\\
=&   \frac{\ud}{\ud t}\int_{a(t)}^{b(t)}  \bbs{e(c) \sqrt{1+h_x^2} }  \ud x  -  \int_{a(t)}^{b(t)} \gamma(c) H h_t  \ud x  -  e(c) I_b + e(c) I_a,
\end{align*}
where we used the Reynolds transport theorem and
\begin{equation}
I_b :=  b'   \sqrt{1+h_x^2}\big|_b +   \frac{h_x h_{t}}{\sqrt{1+h_x^2}}\Big|_b, \quad     I_a:=  a'  \sqrt{1+h_x^2}\big|_a +  \frac{h_x h_{t}}{\sqrt{1+h_x^2}}\Big|_a.
\end{equation}
Then by compatibility condition \eqref{com-com},
\begin{equation}
I_b = b'  \bbs{ \sqrt{1+h_x^2} - \frac{h_x  u_x }{\sqrt{1+h_x^2}}} = \frac{b' (1+ h_x w_x)}{\sqrt{1+h_x^2}}\Big|_b, \quad I_a =  \frac{a' (1+ h_x w_x)}{\sqrt{1+h_x^2}}\Big|_a.
\end{equation}
On the other hand, the right-hand-side of \eqref{ttt} becomes
\begin{align*}
\int_a^b D e'(c) \pt_x \bbs{\frac{c_x}{\sqrt{1+h_x^2}}} \ud x
=- D \int_a^b e''(c) \frac{c_x^2}{\sqrt{1+h_x^2}} \ud x +D \frac{\pt_x e(c)}{\sqrt{1+h_x^2}} \Big|_a^b.
\end{align*}
Therefore,
\begin{equation}\label{energy_c}
\begin{aligned}
& \frac{\ud}{\ud t}\int_{a(t)}^{b(t)}  \bbs{e(c) \sqrt{1+h_x^2} } \ud x -  \int_{a(t)}^{b(t)} \gamma(c) h_t H \ud x + D \int_a^b e''(c) \frac{c_x^2}{\sqrt{1+h_x^2}} \ud x\\
  =& e(c(b)) I_b  - e(c(a)) I_a   +  D \frac{\pt_x e(c)}{\sqrt{1+h_x^2}} \Big|_a^b.
\end{aligned}
\end{equation}

Particularly for $e(c)=c$, we have
\begin{equation}\label{total_flux}
  0= \frac{\ud}{\ud t}\int_{a(t)}^{b(t)}  \bbs{ c \sqrt{1+h_x^2} } \ud x =  \cos \theta_b [ D  c_x|_b+ c(b)(1+h_x w_x)b'] -  \cos \theta_a [ D c_x|_{a}+ c(a)(1+h_x w_x)  a'],
\end{equation}
where $\cos \theta_a = \frac{1}{\sqrt{1+h_x^2}}\Big|_a , \cos \theta_b = \frac{1}{\sqrt{1+h_x^2}}\Big|_b$ and $\theta_a, \theta_b$ are the dynamic contact angle at $a,b$ respectively.
Same as \eqref{3D_bc_c}, to maintain the mass conservation law, we impose the Robin boundary condition for surfactant concentration \eqref{eq-c2D}
\begin{equation}\label{BC_c}
D  c_x|_b+ c(b)(1+h_x w_x)|_b ~b'=0, \quad   D c_x|_{a}+ c(a)(1+h_x w_x)|_a ~  a'=0.
\end{equation}

Using boundary condition \eqref{BC_c}, we further simplify
\eqref{energy_c} as
\begin{equation}\label{energy_c_n}
\begin{aligned}
& \frac{\ud}{\ud t}\int_{a(t)}^{b(t)}  \bbs{e(c) \sqrt{1+h_x^2} } \ud x -  \int_{a(t)}^{b(t)} \gamma(c) h_t H \ud x + D \int_a^b e''(c) \frac{c_x^2}{\sqrt{1+h_x^2}} \ud x\\
  =&  \cos \theta_b \bbs{e(c(b))(1+ h_x w_x)|_b b' + D e' c_x|_b} -  \cos \theta_a \bbs{e(c(a)) (1+ h_x w_x)|_a a' + D e' c_x|_a}\\
  =&\cos \theta_b  \gamma(c(b)) b' (1+ h_x w_x)|_b -  \cos \theta_a  \gamma(c(a)) a' (1+ h_x w_x)|_a.
\end{aligned}
\end{equation}

Second, we compute the rate of change of the total free energy including the total surface energy and the gravitational potential energy.

For a 2D droplet placed on an inclined textured surface, with the gravitational effect and the volume constraint $V$, we take the total free energy  of the droplet as
\begin{equation}\label{energy}
\begin{aligned}
\F(h(t), a(t), b(t), \lambda(t))= &\int_{a(t)}^{b(t)}  e(c) \sqrt{1+ (\pt_x h)^2} \ud x  + (\gamma_{\nsl}-\gamma_{\nsg}) \int_{a(t)}^{b(t)} \sqrt{1+ (\pt_x w)^2} \ud x\\
&+ \rho g  \int_{a(t)}^{b(t)}\int_{w}^{u+w}(y\cos \theta_0+x\sin\theta_0)\ud y \ud x - \lambda(t) \bbs{ \int_{a(t)}^{b(t)} u \ud x - V} ,
\end{aligned}
\end{equation}
where $h= u+ w$, $e(c)$ is the energy density on the capillary surface, $\rho$ is the density of the liquid, and $g$ is  the gravitational acceleration. Denote $\kappa:= \rho g.$

Notice  the time derivative of the second term in $\F$ is
$$
  \frac{\ud}{\ud t} (\gamma_{\nsl}-\gamma_{\nsg}) \int_{a(t)}^{b(t)} \sqrt{1+ (\pt_x w)^2}  \ud x =  (\gamma_{\nsl}-\gamma_{\nsg})\bbs{b'  \sqrt{1+ (\pt_x w)^2}|_b -a' \sqrt{1+ (\pt_x w)^2}|_a},
$$
and from $u|_{a,b}=0$, the time derivative of the third term in $\F$ is
\begin{equation}
  \frac{\ud}{\ud t}  \kappa  \int_{a(t)}^{b(t)}\int_{w}^{u+w}(y\cos \theta_0+x\sin\theta_0)\ud y \ud x =  \kappa  \int_{a(t)}^{b(t)}h_t (h\cos \theta_0+x\sin\theta_0) \ud x. 
\end{equation}
This,  together with the energy dissipation \eqref{energy_c_n}, gives
\begin{equation}\label{2D_rough_E}
\begin{aligned}
  \frac{\ud}{\ud t} \F  
=& -  \int_{a(t)}^{b(t)}  \bbs{ -\gamma(c) H + \lambda - \kappa(h\cos \theta_0+x\sin\theta_0 } h_t \ud x
 - D \int_{a(t)}^{b(t)}  e''(c) \frac{c_x^2}{\sqrt{1+h_x^2}} \ud x \\   
     &+ b' [\cos \theta_b  \gamma(c(b)) (1+ h_x w_x)|_b + (\gamma_{\nsl}-\gamma_{\nsg}) \sqrt{1+ (\pt_x w)^2}|_b] \\
     &- a' [\cos \theta_a  \gamma(c(a))(1+ h_x w_x)|_a + (\gamma_{\nsl}-\gamma_{\nsg})\sqrt{1+ (\pt_x w)^2}|_a].
\end{aligned}
\end{equation}

\subsection{Energetic considerations: the Onsager principle and the governing equations}

Same as Section \ref{sec_2}, we
choose the Rayleigh dissipation functional as
\begin{equation}
Q:=\frac{\beta}2 \int_{a(t)}^{b(t)} \frac{ h_t^2}{\sqrt{1+h_x^2}} \ud x + \frac{\ssr }{2} (|b'|^2+ |a'|^2) +  \frac{D}{2} \int_{a(t)}^{b(t)} e''(c) \frac{c_x^2}{\sqrt{1+h_x^2}} \ud x.
\end{equation}
Then by the same derivations as \eqref{ray}-\eqref{3Dfull} for the 3D case, we conclude 
the governing equations for the full dynamics of a 2D single droplet on a textured substrate 
\begin{equation}\label{eq_full}
\left\{
\begin{aligned}
 & \frac{\beta}{\sqrt{1+h_x^2}} h_t= -\gamma(c) H - \kappa(h\cos \theta_0+x\sin\theta_0)+ \lambda, \quad \text{ in } (a(t),b(t))\\
 & \qquad \qquad \qquad  (h-w)|_a =0, \quad  (h-w)|_b =0,\\
 &  c_t  - v_n \pt_x \bbs{c \frac{h_x }{ \sqrt{1+h_x^2}}}=D \pt_{ss}c, \quad  \text{ in } (a(t),b(t))\\
 & \quad \quad \qquad   D  c_x|_b+ c(b)b'(1+h_x w_x)|_b=0, \quad   D c_x|_{a}+ c(b)  a' (1+h_x w_x)|_a=0,\\
&  \ssr  b'=-\cos \theta_b  \gamma(c(b)) (1+ h_x w_x)|_b - (\gamma_{\nsl}-\gamma_{\nsg}) \sqrt{1+ (\pt_x w)^2}|_b,\\
&   \ssr  a'=  \cos \theta_a  \gamma(c(a))(1+ h_x w_x)|_a + (\gamma_{\nsl}-\gamma_{\nsg})\sqrt{1+ (\pt_x w)^2}|_a,\\
 & \int_{a(t)}^{b(t)} h \ud x = V,
\end{aligned}
\right.
\end{equation}
where $\pt_s = \frac{1}{\sqrt{1+h_x^2}} \pt_x$ and 
$H = -\pt_x \bbs{\frac{h_x}{\sqrt{1+h_x^2}}} = -\frac{h_{xx}}{(1+h_x^2)^{\frac{3}{2}}}$ is the mean curvature. After taking into account the   textured substrate, we remark that the unbalanced Young force  becomes 
\begin{align}
F_b = -\cos \theta_b  \gamma(c(b)) (1+ h_x w_x)|_b - (\gamma_{\nsl}-\gamma_{\nsg}) \sqrt{1+ (\pt_x w)^2}|_b,\\
F_a =\cos \theta_a  \gamma(c(a))(1+ h_x w_x)|_a + (\gamma_{\nsl}-\gamma_{\nsg})\sqrt{1+ (\pt_x w)^2}|_a.
\end{align}

As a consequence, the energy dissipation relation \eqref{ED} becomes
\begin{align*}
\frac{\ud}{\ud t} \F =- \beta \int_{a(t)}^{b(t)} \frac{ h_t^2}{\sqrt{1+h_x^2}}  \ud x
 - D \int_a^b e''(c) \frac{c_x^2}{\sqrt{1+h_x^2}} \ud x - \ssr  (|b'|^2 + |a'|^2) =-2Q.
\end{align*}
Before  proceeding to the computations for the full dynamics \eqref{eq_full}, we recast the  equation \eqref{recast_c2D} for the dynamics of surfactant concertration $c$ as
\begin{equation}\label{num_c}
 c_t  -\frac{h_t h_x}{1+h_x^2} c_x - \frac{h_t h_{xx}}{(1+h_x^2)^2} c =D \frac{c_{xx}}{1+h_x^2} -\frac{D h_x h_{xx}}{(1+h_x^2)^2} c_x,
\end{equation}
which is a computationally friendly form.

\subsection{Proposed numerical scheme}\label{sec_scheme}

We will split the PDE solver for the full dynamics of droplets with surfactant into the following three steps: (i) explicit boundary updates; (ii) semi-implicit capillary surface updates and (iii) implicit surfactant updates.
The unconditional stability for the explicit 1D boundary updates is proved in \cite{gao2020gradient}, which efficiently decouples the computations of the boundary evolution and the capillary surface updates. The semi-implicit  capillary surface updates with the volume constraint and the implicit surfactant updates can be convert to  standard elliptic solvers at each step. 

Let $t^n = n\Delta t$, $n=0, 1, \cdots$ with time step $\Delta t$. We  approximate $a(t^n), b(t^n), h(t^n) $ by $a^n, b^n, h^n$ respectively. We present the first order scheme as follows.  For completeness, we also provide a pseudo-code in Appendix \ref{sec_code}.\\
 \textit{First order scheme:}
 
Step 1.  Explicit boundary updates. Compute the one-side approximated derivative of $h^n$ at $b^n$ and $a^n$, denoted as $(\pt_x h^n)_N$ and $(\pt_x h^n)_0$. Then by  the moving contact line   boundary conditions in \eqref{eq_full}, we update $a^{n+1}, b^{n+1}$ using 
\begin{equation}
\begin{aligned}\label{end_a}
\ssr \frac{a^{n+1}-a^n}{\Delta t}&=
 \cos \theta^n_a  \gamma(c^n_0)(1+ (h^n_x)_0 (w_x)_0) + (\gamma_{\nsl}-\gamma_{\nsg})\sqrt{1+ (\pt_x w)_0^2}, \quad  \cos \theta^n_a = \frac{1}{\sqrt{1+(h_x^n)_0^2}},\\
\ssr \frac{b^{n+1}-b^n}{\Delta t}&= -\cos \theta^n_b  \gamma(c^n_N)(1+ (h^n_x)_N (w_x)_N) - (\gamma_{\nsl}-\gamma_{\nsg})\sqrt{1+ (\pt_x w)_N^2}, \quad  \cos \theta^n_b = \frac{1}{\sqrt{1+(h_x^n)_N^2}}.
\end{aligned}
\end{equation}

Step 2. Rescale $h^n$ from $[a^n, b^n]$ to $[a^{n+1}, b^{n+1}]$ with $O(\Delta t ^2)$ accuracy using an ALE discretization. 
For $x^{n+1}\in[a^{n+1}, b^{n+1}]$, denote the map from moving grids at $t^{n+1}$ to $t^n$ as
\begin{equation}
x^n:= a^n + \frac{b^n-a^n}{b^{n+1}-a^{n+1}}(x^{n+1}-a^{n+1})\in[a^n, b^n].
\end{equation}
Define the rescaled solution for $h^n$ as
\begin{equation}\label{inter-u-0}
    h^{n*}(x^{n+1}):= h^n(x^n)+ \pt_x h^n (x^n)(x^{n+1}-x^n).
\end{equation}
By the Taylor expansion, it is easy to verify   that
$
    h^{n*}(x^{n+1}) = h^n(x^{n+1}) + O(|x^n-x^{n+1}|^2).
$ From \cite[(B.11)]{gao2020gradient}, we have the first order accuracy of
\begin{equation}
\pt_t h (x^{n+1}, t^{n+1}) = \frac{h(x^{n+1}, t^{n+1}) - h^{n*}(x^{n+1}, t^n)}{\Delta t} + O(\Delta t).
\end{equation}

Step 3. Capillary surface updates  with the volume  constraint. Update $h^{n+1}$ and $\lambda^{n+1}$ semi-implicitly.
\begin{equation}\label{tm313}
\begin{aligned}
\frac{\beta}{\sqrt{1+ (\pt_x h^{n*})^2}} \frac{h^{n+1}-h^{n*}}{\Delta t}=   \frac{\gamma(c^n)}{(1+ (\pt_x h^{n*})^2)^{\frac{3}{2}}} \pt_{xx} h^{n+1}-\kappa (h^{n+1}\cos \theta_0 + x^{n+1}\sin \theta_0)+\lambda^{n+1}, &\\
h^{n+1}(a^{n+1})=w(a^{n+1}), \quad h^{n+1}(b^{n+1})=w(b^{n+1})&,\\
\int_{a^{n+1}}^{b^{n+1}} {u}^{n+1}(x^{n+1}) \ud x^{n+1} = \int_{a^0}^{b^0} u^0(x^0) \ud x^0,&
\end{aligned}
\end{equation}
where the independent variable  is $x^{n+1}\in(a^{n+1},b^{n+1})$.

Step 4. Update the concentration of surfactant.
\begin{equation}\label{sur_update}
\begin{aligned}
(h_t)^{n+1}:= \frac{h^{n+1}-h^{n*}}{\Delta t},\\
(1+(h_x^{n+1})^2)\frac{c^{n+1}-c^n}{\Delta t} = D \pt_{xx}c^{n+1} - \frac{D h_x^{n+1} h_{xx}^{n+1}}{1+(h_x^{n+1})^2} \pt_x c^n + h_t^{n+1} h_x^{n+1}\pt_x c^n +   \frac{h_t^{n+1} h_{xx}^{n+1}}{1+(h_x^{n+1})^2}  c^n
\end{aligned}
\end{equation}
with boundary conditions
\begin{equation}
\begin{aligned}
D  (c_x)^{n+1}_0+ c^{n+1}_0(1+(h_x)_0 (w_x)_0)\frac{a^{n+1}-a^n}{\Delta t}=0,\\
D  (c_x)^{n+1}_N+ c^{n+1}_N(1+(h_x)_N (w_x)_N)\frac{b^{n+1}-b^n}{\Delta t}=0.
\end{aligned}
\end{equation}

We remark that the second order numerical scheme developed in \cite{gao2020gradient} can be adapted here. When there are topological changes of droplets such as splitting and merging due to an impermeable textured substrate,   the projection method for solving variational inequalities developed in \cite{gao2020projection} can also be adapted.

\section{Computations for droplets with dynamic surface tension}\label{sec_exm}

We now use the numerical scheme in Section \ref{sec_scheme} to demonstrate several challenging examples:
(i) the  surface tension decreasing phenomena and asymmetric capillary surfaces due to the presence of the surfactant; (ii) the enhanced contact angle hysteresis or resistance with gravity for droplets placed on an inclined substrate;  (iii) droplets on a textured substrate or in a container with different surfactant concentrations.

\subsection{ Surface tension decreasing phenomena and  asymmetric capillary surface due to the presence of surfactant}
In the first example, we compute the spreading of a droplet placed on a flat plane to observe the surface tension decreasing phenomena due to the presence of different concentrations of surfactant. 

First, we set the initial droplet as a spherical cap profile
\begin{equation}\label{initial}
h(x,0) = \sqrt{R^2 - x^2} -R \cos(\theta_{\text{in}}) \quad \text{ with }\,\, R = \frac{b_0}{\theta_{\text{in}}}, \quad  b_0 = 3.7,\quad \theta_{\text{in}} = \frac{3\pi}{16}
\end{equation}
and the computational parameters as follows
\begin{equation}\label{comp_para}
\beta = 0.1, \quad \kappa = 0.5, \quad \gamma_{\nsl} - \gamma_{\nsg}=-0.7, \quad  \ssr =1; \quad D=0.1;\quad  T = 1.5; \quad \Delta t = 0.015, \quad N=800.
\end{equation}
Following  \eqref{Lang}, the surfactant-dependent surface tension
$\gamma(c)$ is taken to be
\begin{equation}
\gamma(c) = \gamma_0 + \ln(1-c) \quad  \text{ with } \gamma_0 = 2.
\end{equation}
This means $\gamma(c)$ is decreasing w.r.t $c$ and if $c=0$, the equilibrium Young's angle is given by $\cos \theta_Y =- \frac{\gamma_{\nsl} - \gamma_{\nsg}}{\gamma_0}=0.35.$

In Fig. \ref{fig_plane} (1st), we first take $c= 0$ and compute the spreading process without the surfactant starting from the initial droplet \eqref{initial} to time $T=1.5$. We observe the initial  symmetric droplet (marked with black pentagrams) tends to shrink to it's equilibrium symmetrically; see the symmetric droplet profile at $T=1.5$ (marked with black circles). Notice the dynamics of the concentration of surfactant $c$ is shown on the capillary surface using different color; see color bar on the right side of the figures. The evolution of the capillary surface is drawn at equal time intervals with solid thin lines and patched with color showing the surfactant concentration. 

Then in Fig. \ref{fig_plane} (2nd),
we take a uniform initial concentration of surfactant $c(x,0)=0.8$  and start from the same initial symmetric droplet, which is marked with black pentagrams and is patched with a uniform color. We observe that, as time increases to $T=1.5$, the droplet tends to spread out  like a thin film due to the lower effective surface tension $\gamma(c)$; see the flatten  droplet profile at $T=1.5$ (marked with black circles). During the spreading process, the concentration of the surfactant at two contact endpoints decreases, so we observe the droplet still holds a pancake shape instead of completely spreading out; see similar droplet profiles in the lubrication model \cite{xu2016variational, limat2004three}.

To see the significant contribution of different surfactant concentrations to the droplet profile,  we use the same initial capillary surface (marked with black pentagrams) but  take an asymmetric initial concentration of the surfactant in Fig. \ref{fig_plane} (3rd). Explicitly, we take initial concentration as
\begin{equation}\label{in_c}
c(x,0) =  0.5+ \frac{0.6}{\pi} \arctan(100 x),
\end{equation}
which increases from $0.2$ to $0.8$ with a sharp transition; see the patched curve marked with  black pentagrams. Then as time increases, the surfactant 'drags' the droplet to the right and induces an asymmetric motion. We can observe an advancing contact angle and a receding contact angle in the droplet profile at $T=1.5$ (marked with black circles). {\blue To further observe the long time behaviors of the droplets, with the same initial asymmetric concentration of surfactant \eqref{in_c}, we compute the dynamics of the droplet profile up to $T=25$. We use same computational parameters as in \eqref{comp_para} except $T=25, \,\Delta t=0.125,\,N=1600.$ The asymmetric droplet profile dragged by surfactant becomes symmetric again with approximated constant-concentration of surfactant at $T=20$. This is a numerical justification for the convergence of dynamic solution to the steady spherical cap solution given in \eqref{capS}. }


\begin{figure}
\includegraphics[scale=0.35]{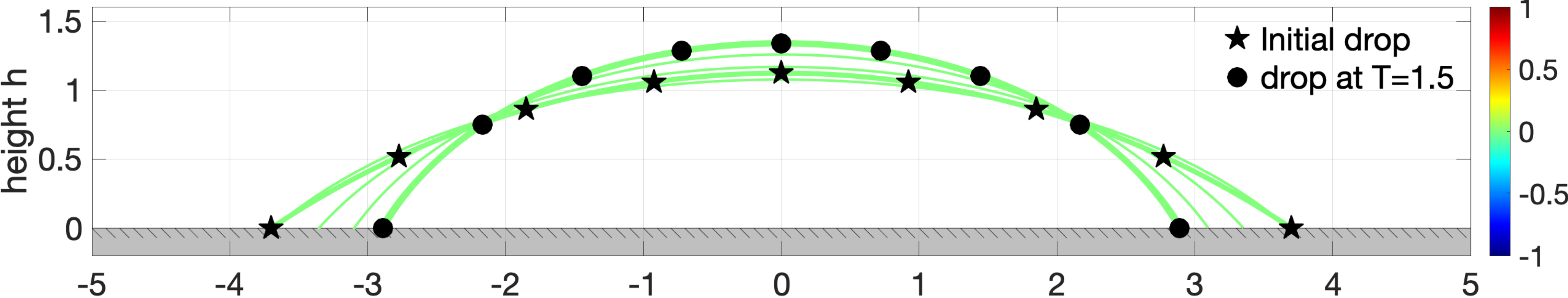} 
\includegraphics[scale=0.35]{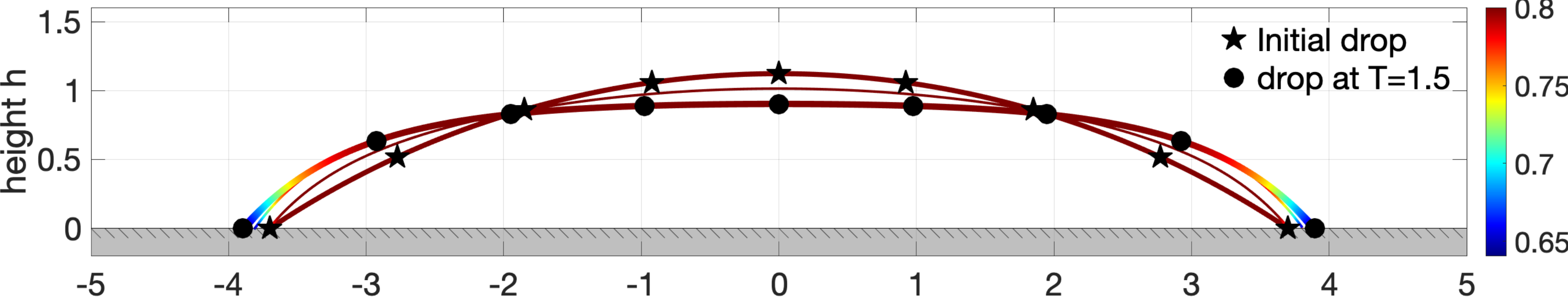}
 \includegraphics[scale=0.35]{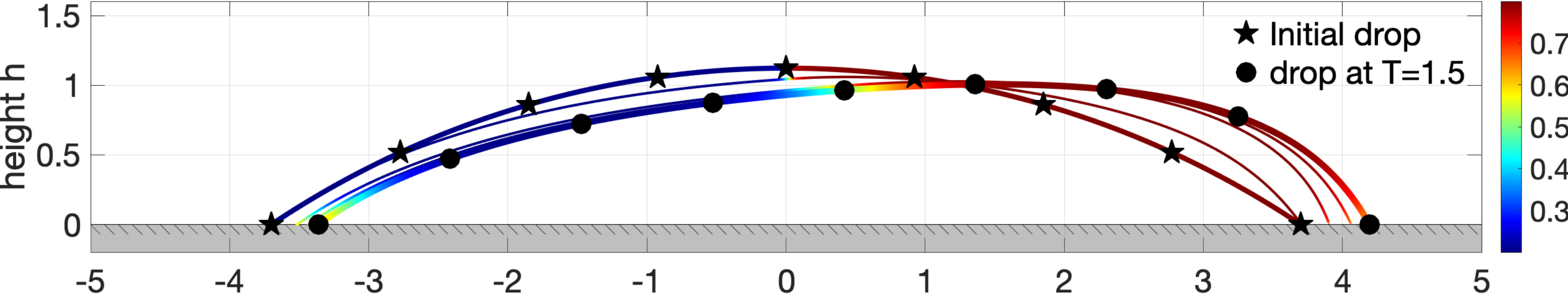} 
 \includegraphics[scale=0.35]{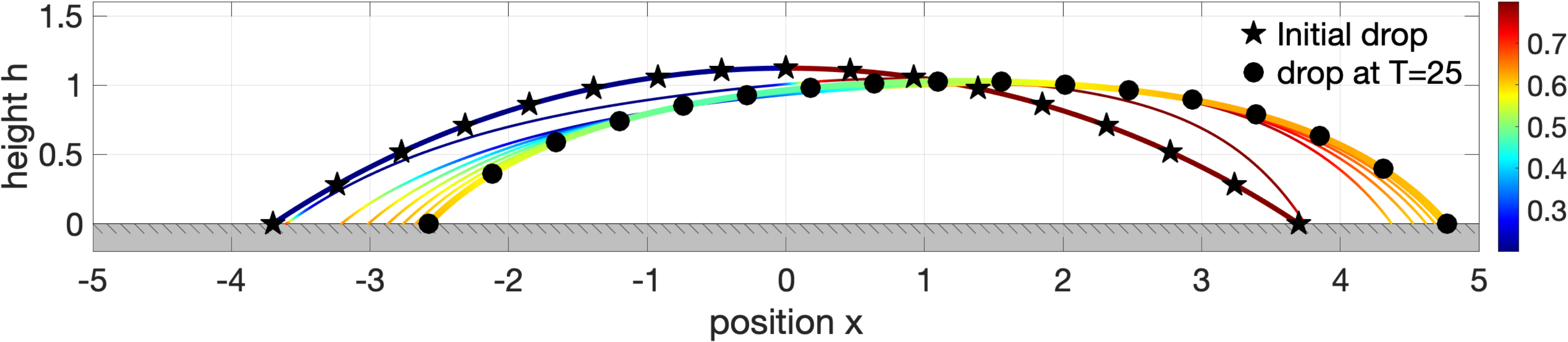} 
 \caption{Spreading of droplets on a plane with lower surface tension due to the surfactant effect. 1st: Time evolution of the capillary surface without the surfactant starting from an initial profile \eqref{initial} marked with black pentagrams to a doplet profile marked with black circles at the final time $T=1.5$. 2nd: With initial concentration $c=0.8$, the flatten droplet with a pancake shape is shown with a density-patched surface at $T=1.5$. 3rd: With an asymmetric  initial concentration \eqref{in_c}, an asymmetric evolution with the surfactant 'drag' effect is shown at equal time intervals and at $T=1.5$.  {\blue 4th: The asymmetric profile dragged by surfactant with  initial concentration \eqref{in_c}  turns out to be symmetric with approximated constant-concentration of surfactant at $T=25$.} }\label{fig_plane}
\end{figure}

\subsection{Droplet on an inclined  surface }
In the second example, we compute the spreading  of a droplet placed on an inclined substrate to observe the competition between the gravitational effect and the surfactant-dependent capillary effect due to the presence of different concentration of the surfactant. 

We use the same initial droplet profile \eqref{initial} and take the inclined angle $\theta_0 = 0.3$ for the substrate.
We use the following computational parameters in Fig. \ref{fig_inc}
\begin{equation}
\beta = 0.1,  \quad \kappa = 0.5, \quad \gamma_{\nsl} - \gamma_{\nsg}=-0.75, \quad  \ssr =1, \quad D=0.1, \quad \Delta t = 0.02, \quad N=800.
\end{equation}
Same as Fig. \ref{fig_plane}, the evolution of the capillary surface is drawn at equal time intervals with solid thin lines and patched with color showing  surfactant concentration. 

In Fig. \ref{fig_inc} (upper one), we take $c=0$ and compute the evolution of a droplet without surfactant as a comparison. We can observe the contact angle hysteresis (CAH) in the droplet profile at $T=2$ (marked with black circles)  with  an advancing contact angle and a receding contact angle due to the gravity, which is  unapparent. 

However, in Fig. \ref{fig_inc} (middle  one), to see the 
 enhanced CAH due to the surfactant effect, we take an asymmetric initial concentration 
\begin{equation}\label{in_c_left}
c(x,0) =  0.45- \frac{0.5}{\pi} \arctan(100 x),
\end{equation}
which decreases from $0.7$ to $0.2$ with a sharp transition; see the patched curve marked with  black pentagrams in Fig. \ref{fig_inc} (middle  one). As time increases, we see the gravity  and the asymmetric concentration of the surfactant (left part higher than the right part of the capillary surface) accelerate the rolling down of the droplet. Then droplet profile at $T=2$ is marked with black circles and patched with the surfactant concentration, in which we observe a significant enhancement of CAH phenomena with  very different advancing contact angles and  receding contact angles. On the other hand, if we switch the initial concentration of the surfactant to 
\begin{equation}\label{in_c_right}
c(x,0) =  0.45+ \frac{0.5}{\pi} \arctan(100 x)
\end{equation}
so that the right part has higher concentration than the left part of the capillary surface. Then in Fig. \ref{fig_inc} (lower  one),  we observe the surfactant effect wins the competition with the gravity and the droplet even rises up instead of rolling down; see  the  droplet profile at $T=2$ (marked with black circles).


\begin{figure}
\includegraphics[scale=0.34]{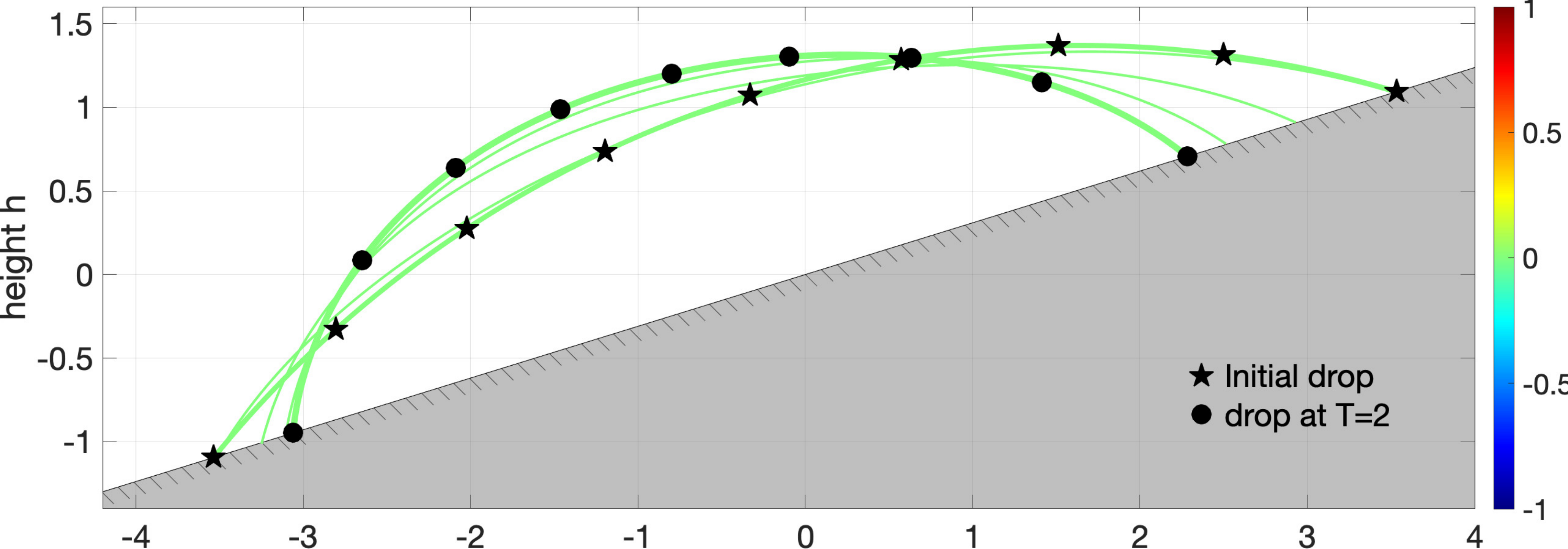}
\includegraphics[scale=0.34]{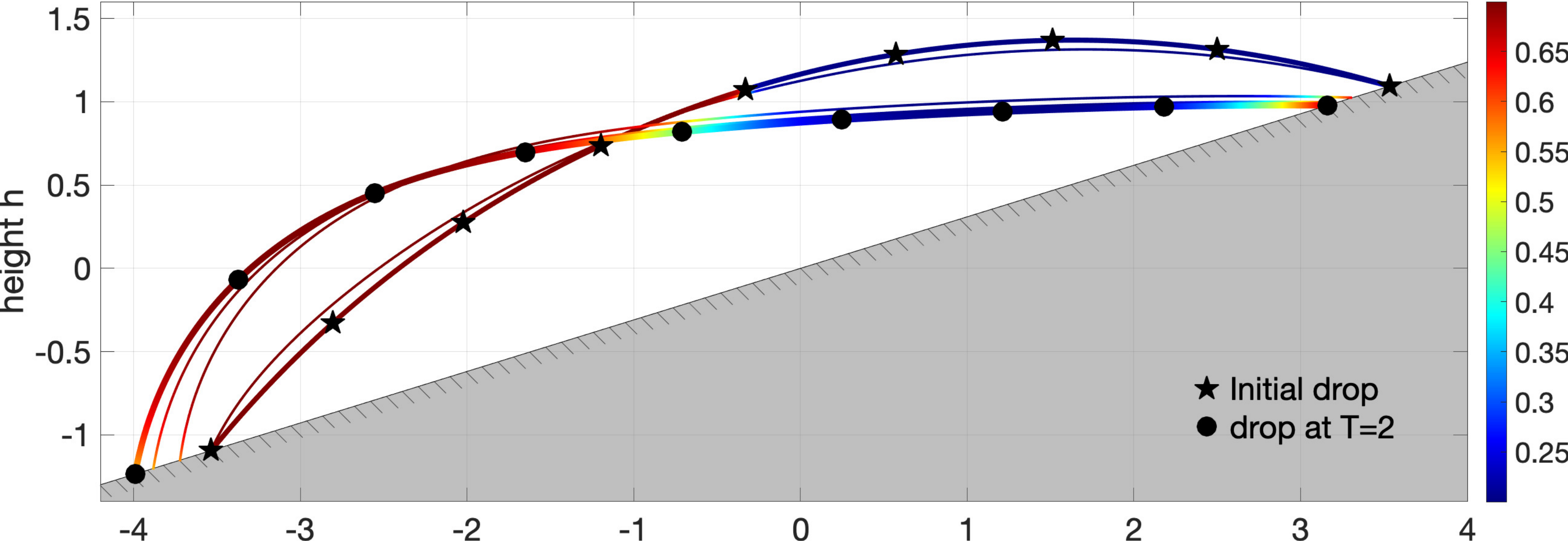} 
 \includegraphics[scale=0.34]{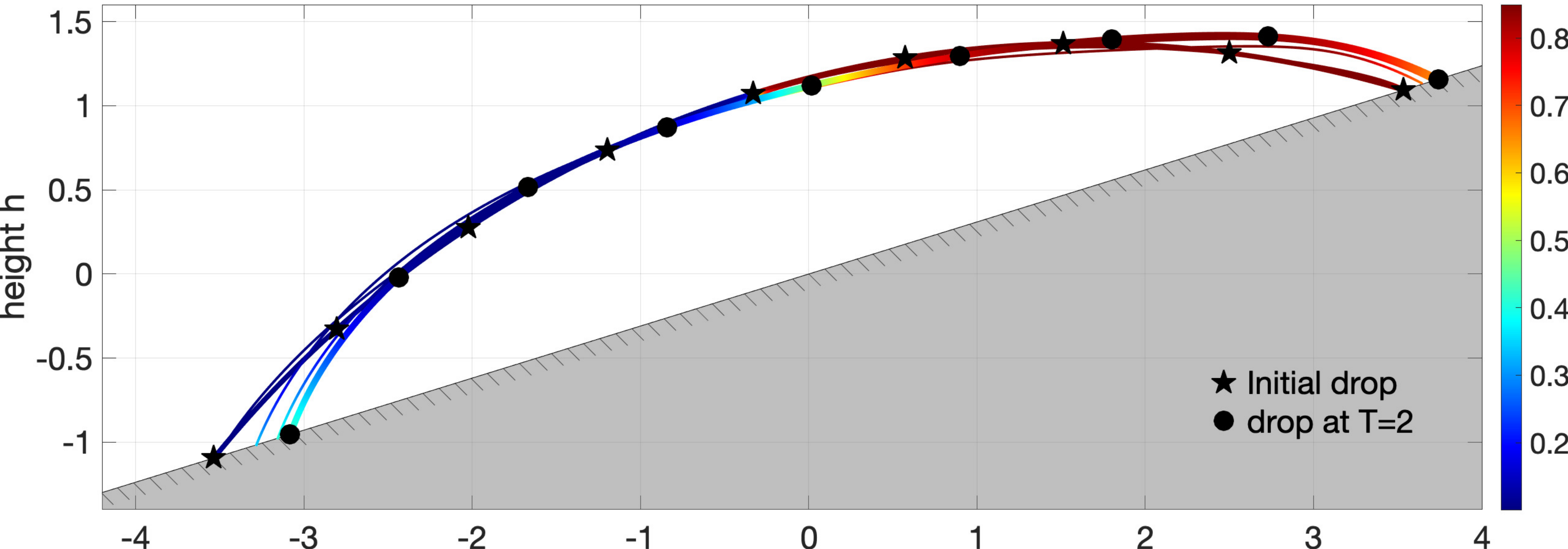} 
 \caption{Enhanced CAH and resistance with the gravity due to asymmetric concentration of surfactant for a droplet placed on an inclined substrate with $\theta_0=0.3$.
 (upper) Time evolution of the capillary surface without the surfactant starting from an initial profile \eqref{initial} marked with black pentagrams to a droplet profile marked with black circles at a final time $T=2$. (middle) With the initial concentration \eqref{in_c_left}, the significant enhancement of rolling down and the CAH phenomena are shown with  density-patched surfaces at equal time intervals and at $T=2$. (lower) With the initial concentration \eqref{in_c_right}, the droplet rises up instead of rolling down because the surfactant effect wins the competition with the gravity.  
 }\label{fig_inc}
\end{figure}

\subsection{Droplets on a textured substrate and  in a cocktail glass}
In the third example, we compute the spreading  of a droplet  on some typical textured substrates such as a cocktail glass and a substrate with constantly changed effective slope. 
The common computational parameters are 
\begin{equation}
\beta = 0.1, \quad \kappa = 0.5,  \quad  \ssr =1; \quad D=0.5; \quad \Delta t = 0.02.
\end{equation}

In Fig. \ref{fig_rough} (upper), we take the initial droplet profile (marked with black pentagrams) as
\begin{equation}\label{initial_rough}
h(x,0) = \sqrt{R^2 - x^2} -R \cos(\theta_{\text{in}})+ w(b_0) + \frac{[w(b_0)-w(-b_0)](x+b_0)}{2b_0}, \quad \,\, R = \frac{b_0}{\theta_{\text{in}}}, \quad  b_0 = 3.7
\end{equation}
with $ \theta_{\text{in}} = \frac{3\pi}{16}$ and  a cocktail glass substrate
\begin{equation}
w(x) = 0.5 \sqrt{x^2+0.1}.
\end{equation} 
Then taking $\gamma_{\nsl} - \gamma_{\nsg}=-0.9, \, N=1600$ and the initial concentration of the surfactant as $c(x,0)=0.2$,  time evolution of the density-patched capillary surfaces is shown at equal time intervals and at the final time $T=4$ (marked with black circles). We observe the capillary rise near the contact lines and the surfactant tends to push themselves and concentrate near the contact lines.

In Fig. \ref{fig_rough} (lower), we take the initial droplet profile (marked with black pentagrams) as
\eqref{initial_rough}
with $ \theta_{\text{in}} = \frac{1.3\pi}{8}$, an effective inclined angle $\theta_0=0.2$ and  a textured substrate
\begin{equation}\label{w_t}
w(x) = 0.1 \bbs{\sin(2x)+\cos(4x)}^2.
\end{equation} 
Then taking $\gamma_{\nsl} - \gamma_{\nsg}=-0.5, \, N=800$ and the initial asymmetric concentration of the surfactant as 
\begin{equation}\label{in_c_right_n}
c(x,0) =  0.45+ \frac{0.7}{\pi} \arctan(100 x),
\end{equation}
time evolution of the density-patched capillary surfaces is shown at equal time intervals and at final time $T=2$ (marked with black circles). We observe an asymmetric rising up of the droplet due to the asymmetric initial concentration and the constantly changed effective slope of the textured substrate.

\begin{figure}
\includegraphics[scale=0.35]{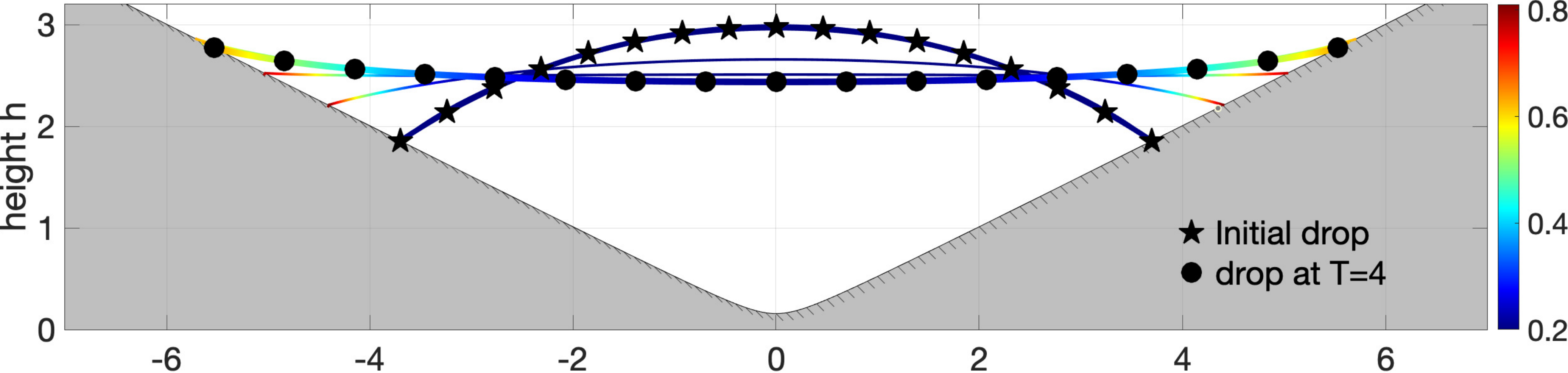} 
\medskip

\includegraphics[scale=0.35]{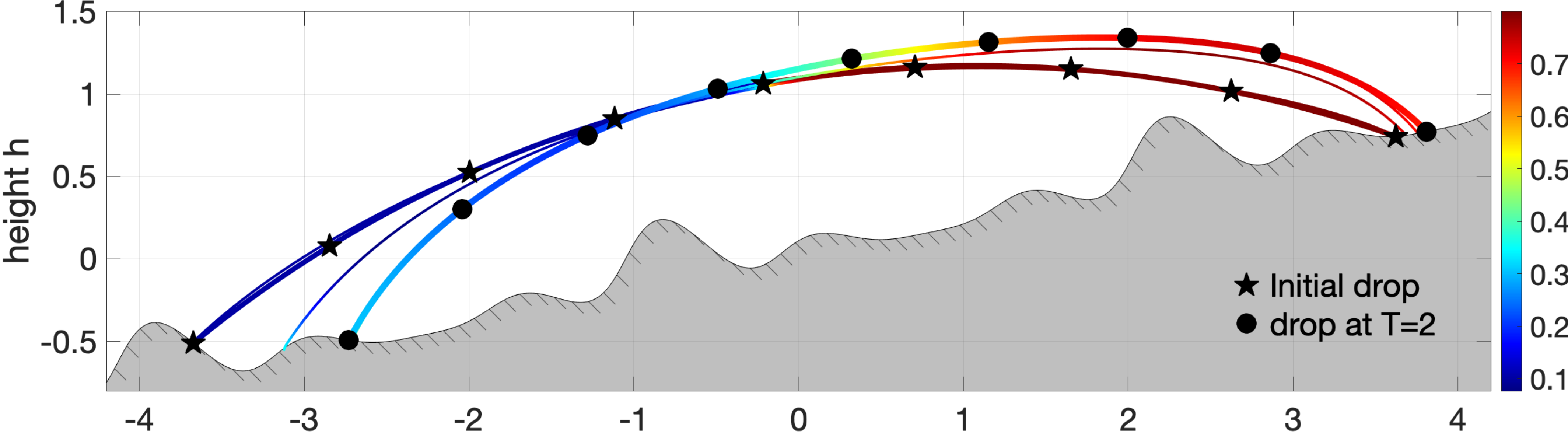} 
\caption{The upper figure is the time evolution of a droplet in a cocktail glass with a uniform initial concentration of surfactant. The initial profile \eqref{initial_rough} is marked with black pentagrams while the density-patched capillary surface at the final time $T=4$ is marked with black circles. The lower figure is  the time evolution of a droplet on a inclined textured substrate   \eqref{w_t}. An asymmetric rising up starting from \eqref{in_c_right_n} is shown with the density-patched capillary surfaces at equal time intervals and at  $T=2$. 
}\label{fig_rough}
\end{figure}

\appendix
\section{Proof of Proposition \ref{prop_equiv}}\label{appA}
\begin{proof}[Proof of Proposition \ref{prop_equiv}]
Recall the continuity equation \eqref{eq-con} on wetting domain $\Omega_t$ and the relation \eqref{rel_c} between $c$ and $\cc$.
We then have
\begin{equation}\label{tm_con}
\begin{aligned}
0=&\pt_t  \bbs{c \sqrt{1+ |\nabla h|^2}} + \nabla \cdot \bbs{c \sqrt{1+|\nabla h|^2} v_{xy} }\\
=&  \sqrt{1+ |\nabla h|^2}  \bbs{(\pt_t + v_{xy} \cdot \nabla_{xy}) c} + c \bbs{\pt_t \sqrt{1+|\nabla h|^2} + \nabla_{xy} \cdot \bbs{\sqrt{1+|\nabla h|^2} v_{xy}}}\\
=&  \sqrt{1+ |\nabla h|^2}  \bbs{(\pt_t + v  \cdot \nabla) \cc} + \cc \bbs{\pt_t \sqrt{1+|\nabla h|^2} + \nabla_{xy} \cdot \bbs{\sqrt{1+|\nabla h|^2} v_{xy}}}
\end{aligned}
\end{equation}
where 
$v_{xy}:= \left( \begin{array}{c}
v_1\\
v_2
\end{array} \right)=\left( \begin{array}{c}
\frac{-h_t h_x}{1+|\nabla h|^2}\\
\frac{-h_t h_y}{1+|\nabla h|^2}
\end{array} \right)
+  \left( \begin{array}{c}
f\\
g
\end{array} \right)$
due to \eqref{flowmap}. 

First, for the last term in \eqref{tm_con}, using the identity  \eqref{MC_ID}, we have
\begin{equation}
\begin{aligned}
&\pt_t \sqrt{1+|\nabla h|^2} + \nabla_{xy} \cdot \bbs{\sqrt{1+|\nabla h|^2} v_{xy}}= h_t H   +  \nabla_{xy} \cdot \bbs{  \sqrt{1+|\nabla h|^2}\left( \begin{array}{c}
f\\
g
\end{array} \right)}
\end{aligned}
\end{equation}
This, together with \eqref{tm_con}, gives the equation for $\cc$
\begin{equation}
  \bbs{(\pt_t + v  \cdot \nabla) \cc} + \cc v_n  H + \frac{1}{\sqrt{1+|\nabla h|^2}} \nabla_{xy} \cdot \bbs{  \sqrt{1+|\nabla h|^2}\left( \begin{array}{c}
f\\
g
\end{array} \right)} = 0.
\end{equation}

Second, we prove the following claim
\begin{equation}\label{tm_second}
 \nabla_s \cdot v_s = \frac{1}{\sqrt{1+|\nabla h|^2}} \nabla_{xy} \cdot \bbs{  \sqrt{1+|\nabla h|^2}\left( \begin{array}{c}
f\\
g
\end{array} \right)} = \pt_x f + \pt_y g + \frac12\left( \begin{array}{c}
f\\
g
\end{array} \right) \cdot \frac{\nabla_{xy}(|\nabla h|^2)}{1+|\nabla h|^2}.
\end{equation}
Denote 
$$\tilde{f}(x,y,h(x,y,t)) = f(x,y,t), \quad \tilde{g}(x,y,h(x,y,t)).$$
Then for  the tangential velocity 
\begin{equation}
v_s = f \tau_1 + g \tau_2 = \left( \begin{array}{c}
\tilde{f}\\
\tilde{g}\\
h_x \tilde{f} + h_y \tilde{g}
\end{array} \right),
\end{equation}
 by the chain rule, we have
\begin{equation}\label{A6}
\nabla \cdot v_s = \pt_x f + \pt_y g.
\end{equation}
On the other hand,
\begin{equation}
\begin{aligned}\label{A7}
-n(n \cdot \nabla) v_s =& -\frac{1}{1+|\nabla h|^2} \left( \begin{array}{c}
-h_x\\
-h_y\\
1
\end{array} \right) \cdot (-h_x \pt_x - h_y \pt_y + \pt_z) \left( \begin{array}{c}
\tilde{f}\\
\tilde{g}\\
h_x \tilde{f} + h_y \tilde{g}
\end{array} \right)\\
=&\frac{1}{1+|\nabla h|^2} \bbs{f(h_{xx}+h_{xy}) + g(h_{xy}+h_{yy})} =  \frac12\left( \begin{array}{c}
f\\
g
\end{array} \right) \cdot \frac{\nabla_{xy}(|\nabla h|^2)}{1+|\nabla h|^2}.
\end{aligned}
\end{equation}
Combining \eqref{A6} and \eqref{A7} yields \eqref{tm_second}.
\end{proof}

\section{Pseudo-code for first order scheme}\label{sec_code}

Below, we present a pseudo-code for the first order scheme in Section \ref{sec_scheme}.
\\1. Grid for time:
$t^n = n \Delta t$, $n=0 ,1, \cdots,$ where $\Delta t$ is time step.
\\ 2. Fix $N$ and set moving grids for space: $x_{j}^n = a^n + {j}{\tau^n},\,\,\tau ^n=\frac{b^n-a^n}{N}$, $j=-1,0, 1, \cdots, N+1.$ 
\\ 3. Calculate volume $V:= \sum_{j=1}^{N-1}( h^0-w)(x^0_j) \tau^0$.
\\4. Denote the finite difference operators
\begin{equation}\label{tm_C2}
\begin{aligned}
&(\pt_x h)_0^n= \frac{4 h_1^n - h_2^n-3h_0^n}{2 \tau^n}, \quad  (\pt_x h)_N^n= \frac{-4 h_{N-1}^n + h_{N-2}^n+ 3 h_{N}^n}{2 \tau^n}, \\
 &(\pt_x h)_j^{n} = \frac{h_{j+1}^{n} - h^{n}_{j-1}}{2 \tau^{n}}, \quad   (\pt_{xx} h)_j^{n}=  \frac{h^{n}_{j+1}-2h^{n}_j + h_{j-1}^{n}} {(\tau^{n})^2},\,\, j=1, \cdots, N-1.
 \end{aligned}
\end{equation}
Denote 
\begin{align*}
(\pt_x w)_0:=\pt_x w(x^n_0), \quad (\pt_x w)_N:=\pt_x w(x^n_N),\quad 
\gamma_i^n:= \gamma(c_i^n), \quad i=0, \cdots, N.
\end{align*}
\\ 5. Update $a^{n+1}, b^{n+1}$, $j=0, 1, \cdots, N$,
\begin{align*}
\ssr \frac{a^{n+1}-a^n}{\Delta t}&=\gamma_0^n\frac{1+ (\pt_x h^n)_0(\pt_x w)_0 }{\sqrt{1+ (\pt_x h^n)_0^2}} + (\gamma_{\nsl}-\gamma_{\nsg}) \sqrt{1+ (\pt_x w)_0^2},
\\
\ssr \frac{b^{n+1}-b^n}{\Delta t}&=-\gamma_N^n\frac{1+ (\pt_x h^n)_N(\pt_x w)_N }{\sqrt{1+ (\pt_x h^n)_N^2}}  -(\gamma_{\nsl}-\gamma_{\nsg}) \sqrt{1+ (\pt_x w)_N^2}.
\end{align*}
\\6. Update the moving grids
$\, x_{j}^{n+1} = a^{n+1} + {j}{\tau^{n+1}}, \,\, \tau^{n+1}= \frac{b^{n+1}-a^{n+1}}{N}, \quad  j=0, 1, \cdots, N.$
\\ 7. From \eqref{inter-u-0}, 
$\,h_j^{n*}= h_j^n + (\pt_x h^n)_j (a^{n+1}-a^n+j(\tau^{n+1}-\tau^n)), \quad j=0, \cdots, N.$
\\ 8. Solve $h^{n+1}$ semi-implicitly
\\For $j=1, \cdots, N-1$,  denote $ \alpha_j=1+ ( h_x^{2})_j^n$ and solve
\begin{align}\label{code-eq-r-0}
&
\beta \alpha_j
\frac{h_j^{n+1}-h^{n*}_j}{\Delta t}
= \gamma_j^n
\frac{h^{n+1}_{j+1} -2 h^{n+1}_j +h^{n+1}_{j-1} }{(\tau^{n+1})^2} -\kappa \alpha_j^{3/2} (h_j^{n+1} \cos \theta_0 + x^{n+1}_j \sin \theta_0  ) +\lambda^{n+1} \alpha_j^{3/2} ,\\
&\sum_{j=1}^{N-1} (h^{n+1}_j-w(x_j^{n+1})) \tau^{n+1} = V, \nonumber
\end{align}
with the Dirichlet boundary condition $h^{n+1}_0= w(x^{n+1}_0),\,\, h^{n+1}_N = w(x_N^{n+1}) $.

Denote a positive-definite matrix $A_{(N-1)\times(N-1)}=(a_{ij})$
with
\begin{equation}\label{tri-r-0}
\begin{array}{c}
a_{j, j-1} := -\gamma_j^n, \,\, a_{j,j+1} := -\gamma_j^n, \quad  
a_{j,j} :=2 \gamma_j^n+ \frac{\beta (\tau^{n+1})^2}{\Delta t}\alpha_j+ \kappa \cos \theta_0 (\tau^{n+1})^2 \alpha_j^{\frac32}
\end{array}
\end{equation}
and a vector of length $N-1$ 
$$
\tilde{f}_j:= \frac{\beta (\tau^{n+1})^2}{\Delta t} h_j^{n*} \alpha_j-\kappa \sin \theta_0 x_{j}^{n+1} (\tau^{n+1})^2 \alpha_j^{\frac32}, \quad j= 1, \cdots, N-1
$$
 and \eqref{code-eq-r-0} becomes for $j=1, \cdots, N-1$, 
\begin{equation}\label{1st-eq-u}
\begin{aligned}
a_{j,j-1} h^{n+1}_{j-1} + a_{j,j} h^{n+1}_j + a_{j,j+1} h^{n+1}_{j+1} - \alpha_j^\frac32 (\tau^{n+1})^2\lambda^{n+1}=\tilde{f}_j.
\end{aligned}
\end{equation}
Denote
\begin{equation}
f_1=\tilde{f}_1+\gamma_1^n w(x_0^{n+1}),\,\,   \{f_j= \tilde{f}_j\}_{j=2}^{N-2} , \,\, f_{N-1} = \tilde{f}_{N-1} + \gamma_{N-1}^n w(x_N^{n+1}), \,\, f_N:= \sum_{j=1}^{N-1} w(x_j^{n+1})  + \frac{V}{\tau^{n+1}}.
\end{equation}
 The resulting linear system $\bar{A} y = f$ has a non-singular   matrix 
$
\bar{A}=\left(
\begin{array}{cc}
A & \alpha^\frac32 \\ 
e^\top & 0
\end{array} 
 \right)_{N\times N},
$
where  $y^\top=(h^{n+1}_1, \cdots, h^{n+1}_{N-1}, -(\tau^{n+1})^2\lambda^{n+1})$ and $e^\top=(1,\cdots , 1)\in \mathbb{R}^{N-1}$.   
\\ 9. Solve $c^{n+1}$ from \eqref{sur_update} implicitly.

 Denote
$$
(h_t)_j^{n+1}:= \frac{h_j^{n+1}-h_j^{n*}}{\Delta t},\quad j = 1, \cdots, N-1,
$$
$$
\tilde{f}_j: =- D \frac{(h_x)_j^{n+1} (h_{xx})_j^{n+1}}{1+(h_x^2)_j^{n+1}} (c_x)_j^n 
+ (h_t)_j^{n+1} (h_x)_j^{n+1}  (c_x)_j^n 
+   \frac{(h_t)_j^{n+1} (h_{xx})_j^{n+1}}{1+(h_x^2)_j^{n+1}}  c_j^n. 
$$
Then \eqref{sur_update} becomes
\begin{equation}\label{ceq-tm15}
\begin{aligned}
,\\
(1+(h_x^2)_j^{n+1}) \frac{c^{n+1}_j-c^n_j}{\Delta t} 
=D \frac{c^{n+1}_{j+1} -2 c^{n+1}_j  + c^{n+1}_{j-1} }{(\tau^{n+1})^2} 
+\tilde{f}_j, \quad j = 1, \cdots, N-1
\end{aligned}
\end{equation}
with boundary conditions
\begin{equation}\label{BC416}
\begin{aligned}
D  (c_x)^{n+1}_0+ c^{n+1}_0(1+(h_x)_0 (w_x)_0)\frac{a^{n+1}-a^n}{\Delta t}=0,\\
D  (c_x)^{n+1}_N+ c^{n+1}_N(1+(h_x)_N (w_x)_N)\frac{b^{n+1}-b^n}{\Delta t}=0,
\end{aligned}
\end{equation}
where
$$
(\pt_x c)_0^{n+1}= \frac{4 c_1^{n+1} - c_2^{n+1} -3c_0^{n+1}}{2 \tau^{n+1}}, \quad  
(\pt_x c)_N^{n+1}= \frac{-4 c_{N-1}^{n+1} + c_{N-2}^{n+1}+ 3 c_{N}^{n+1}}{2 \tau^{n+1}}.
$$
Let
$$
\iota_0: = \frac{2 \tau^{n+1}}{D} (1+(h_x)_0 (w_x)_0)\frac{a^{n+1}-a^n}{\Delta t}, \quad  \iota_N: = \frac{2 \tau^{n+1}}{D} (1+(h_x)_N (w_x)_N)\frac{b^{n+1}-b^n}{\Delta t},
$$
then boundary condition \eqref{BC416} becomes
$$
 (\iota_0 -3)c_0^{n+1} + 4 c_1^{n+1} - c_2^{n+1} = 0,\quad  c_{N-2}^{n+1} - 4 c_{N-1}^{n+1}+ (\iota_N +3)c_N^{n+1}   = 0.
$$
Now we recast  \eqref{ceq-tm15} in a $N+1$ order matrix form 
$$
  B c^{n+1} = f, \quad B=(b_{ij})_{i,j=0, \cdots, N}
$$
where, for $i=0$,
$
b_{00} = \iota_0 -3, \quad b_{01} = 4, \quad b_{02} = -1;
$
for $i=1, \cdots, N-1$, 
$
b_{i, i-1} = -1, \quad
b_{i, i} = 2+ \frac{(\tau^{n+1})^2}{D \Delta t} (1+(h_x^2)_i^{n+1}), \quad
b_{i, i+1} = -1;
$
for $i=N$,
$
 b_{N,N-2}=1,  b_{N,N-1}=-4,  b_{N,N}=\iota_N +3,
$
and
$f_0=0, f_N = 0, f_i =\frac{(\tau^{n+1})^2}{D}\bbs{ \tilde{f}_i + \frac{1+(h_x^2)_i^{n+1}}{\Delta t}
c^n_j} ,  \,\, i=1, \cdots, N-1.$

\section*{Acknowledgments}
The authors would like to thank Prof. Masao Doi and Prof.  Tom Witelski for some helpful suggestions. 
J.-G. Liu was supported in part by the National Science Foundation (NSF) under award DMS-2106988.

\bibliographystyle{abbrv}
\bibliography{dropsur}

\end{document}